\pdfoutput = 1
\documentclass[a4paper,UKenglish,numberwithinsect]{lipics-v2019}
\usepackage{tcolorbox}
\usepackage{hyperref}
\usepackage{tikz}
\usepackage{xspace}
\usepackage{csquotes}
\usepackage{longtable}
\usepackage{amsmath}
\usepackage[normalem]{ulem}
\usepackage{mathtools}
\usepackage{todonotes}
\usepackage[linesnumbered,lined, ruled, noend]{algorithm2e}

\nolinenumbers

\makeatletter
\newcommand{\algorithmfootnote}[2][\footnotesize]{%
  \let\old@algocf@finish\@algocf@finish% Store algorithm finish macro
  \def\@algocf@finish{\old@algocf@finish% Update finish macro to insert "footnote"
    \leavevmode\rlap{\begin{minipage}{\linewidth}
    #1#2
    \end{minipage}}%
  }%
}
\makeatother

\graphicspath{{./images/}}%helpful if your graphic files are in another directory

\newcommand{\fvsfull}{\textsc{Feedback Vertex Set}\xspace}
\newcommand{\fvs}{\textsc{FVS}\xspace}

\newcommand{\OO}{\mathcal{O}}
\newcommand{\cO}{\mathcal{O}}
\newcommand{\tw}{\mathbf{tw}}
\theoremstyle{plain}

\newcommand{\defparprob}[4]{
\begin{tcolorbox}[colback=gray!5!white,colframe=gray!75!black]
  \begin{tabular*}{\textwidth}{@{\extracolsep{\fill}}lr} #1  & {\bf{Parameter:}} #3 \\ \end{tabular*}
  {\bf{Input:}} #2  \\
  {\bf{Question:}} #4
\end{tcolorbox}
}

\hideLIPIcs{}

\title{Improved FPT Algorithms for Deletion to Forest-like Structures}
\author{Kishen N. Gowda}{ IIT Gandhinagar, India,  kishen.gowda@iitgn.ac.in}{}{}{}
\author{Aditya Lonkar}{IIT Madras, India, laditya1235@gmail.com}{}{}{}%mandatory, please use full name; only 1 author per \author macro; first two parameters are mandatory, other parameters can be empty.
\author{Fahad Panolan}{Department of Computer Science and Engineering, IIT Hyderabad, India, fahad@cse.iith.ac.in}{}{}{}%mandatory, please use full name; only 1 author per \author macro; first two parameters are mandatory, other parameters can be empty.
\author{Vraj Patel}{IIT Gandhinagar, India, vraj.patel@iitgn.ac.in}{}{}{}%mandatory, please use full name; only 1 author per \author macro; first two parameters are mandatory, other parameters can be empty.
\author{Saket Saurabh}{Institute of Mathematical Sciences, Chennai, India, saket@imsc.res.in}{}{}{}%mandatory, please use full name; only 1 author per \author macro; first two parameters are mandatory, other parameters can be empty.
\authorrunning{K. Gowda, A. Lonkar, F. Panolan, V. Patel, S. Saurabh}
\Copyright{John Q. Public and Joan R. Public}%TODO mandatory, please use full first names. LIPIcs license is "CC-BY";  http://creativecommons.org/licenses/by/3.0/


\EventEditors{John Q. Open and Joan R. Access}
\EventNoEds{2}
\EventLongTitle{42nd Conference on Very Important Topics (CVIT 2016)}
\EventShortTitle{CVIT 2016}
\EventAcronym{CVIT}
\EventYear{2016}
\EventDate{December 24--27, 2016}
\EventLocation{Little Whinging, United Kingdom}
\EventLogo{}
\SeriesVolume{42}
\ArticleNo{23}
\begin{document}

\maketitle

\begin{abstract}
The {\sc Feedback Vertex Set} problem is undoubtedly one of the most well-studied problems in Parameterized Complexity. In this problem, given an undirected graph $G$ and a non-negative integer $k$, the objective is to test whether there exists a subset $S\subseteq V(G)$ of size at most $k$ such that $G-S$ is a forest. After a long line of improvement, recently, Li and Nederlof [SODA, 2020] 
designed a randomized  algorithm for the problem running in time $\mathcal{O}^{\star}(2.7^k)$\footnote{Polynomial dependency on the input size is hidden in $\mathcal{O}^{\star}$ notation.}. In the Parameterized Complexity literature, several problems around {\sc Feedback Vertex Set} have been studied. Some of these include {\sc Independent Feedback Vertex Set} (where the set $S$ should be an independent set in $G$), {\sc Almost Forest Deletion} and {\sc  Pseudoforest Deletion}. In {\sc  Pseudoforest Deletion}, each connected component in $G-S$ has at most one cycle in it. However, in {\sc Almost Forest Deletion}, the input is a graph $G$ and non-negative integers $k,\ell \in {\mathbb N}$, and the objective is to test whether there exists a vertex subset $S$ of size at most $k$, such that $G-S$ is $\ell$ edges away from a forest. 
 In this paper, using the methodology of Li and Nederlof [SODA, 2020], we obtain the current fastest algorithms for all these problems. In particular we obtain following randomized algorithms.
% for the following problems.  
\begin{enumerate}
\item {\sc Independent Feedback Vertex Set} can be solved in time  $\mathcal{O}^{\star}(2.7^k)$.
\item {\sc  Pseudo Forest Deletion} can be solved in time $\mathcal{O}^{\star}(2.85^k)$.
\item {\sc Almost Forest Deletion}  can be solved in 
 $\mathcal{O}^{\star}(\min\{2.85^k \cdot 8.54^\ell,2.7^k \cdot 36.61^\ell,3^k \cdot 1.78^\ell\})$.
\end{enumerate}

\end{abstract}

\section{Introduction}\label{sec:intro}
%!TEX root = main.tex

\fvsfull (\fvs) is a classical NP-complete
problem and has been extensively studied in all subfields of algorithms
and complexity. In this problem we are
given an undirected graph~$G$ and a non-negative integer~$k$ as input,
and the goal is to check whether there exists a subset~$S\subseteq V(G)$ ({\em  called feedback vertex set or in short fvs}) 
of size at most~$k$ such that~$G-S$
is a forest. This problem originated in combinatorial circuit design
and found its way into diverse applications such as deadlock
prevention in operating systems, constraint 
satisfaction and Bayesian inference in artificial intelligence. 
We refer to the survey by Festa et al.~\cite{FestaPR1999} 
for further details on the algorithmic study of feedback set problems
in a variety of areas like approximation algorithms, linear programming and polyhedral combinatorics.

\fvs has been extensively studied in Parameterized Algorithms. \fvs has played a pivotal role in the development of the field of Parameterized Complexity. The earliest known FPT algorithms for \fvs go back to the late 80s and the early 90s~\cite{Bodlaender1992,DowneyFellows1992} and used the
seminal Graph Minor Theory of Robertson and Seymour. These algorithms are quite impractical because
of large hidden constants in the run-time expressions.
Raman et al.~\cite{RamanSS02} designed an
algorithm with running time $\cO^\star(2^{\OO(k\log\log k)})$ which
basically branched on short cycles in a bounded search tree approach. For \fvs, the first deterministic $\cO^\star(c^k)$ algorithm was designed only in 2005; independently by Dehne et al.~\cite{DBLP:journals/mst/DehneFLRS07} and Guo et al.~\cite{DBLP:journals/jcss/GuoGHNW06}. It is important to note here that a randomized algorithm for \fvs with running time $\cO^\star(4^k )$ was known in as early as 1999~\cite{DBLP:journals/jair/BeckerBG00}. 
The deterministic algorithms led to the race of improving the base of the exponent for \fvs algorithms and several algorithms~\cite{DBLP:conf/soda/Cao18,DBLP:journals/algorithmica/CaoC015,DBLP:journals/jcss/ChenFLLV08,CyganNPPRW11,iwata19,KociumakaP14,li-soda20}, both deterministic and randomized, have been designed. Until few months ago the best known deterministic algorithm for \fvs ran in time $\cO^\star(3.619^k)$~\cite{KociumakaP14}, while the \emph{Cut \& Count} technique by Cygan et al.~\cite{CyganNPPRW11} gave the best known randomized algorithm running in time $\cO^\star(3^k)$. However, just in last few months both these algorithms have been improved; Iwata and Kobayashi~\cite[IPEC 2019]{iwata19} designed the fastest known deterministic algorithm with running time $\cO^\star(3.460^k)$ and Li and Nederlof~\cite[SODA 2020]{li-soda20}  designed the fastest known randomized algorithm with running time $\cO^\star(2.7^k)$. The success on \fvs has led to the study of many variants of \fvs in literature such as {\sc Connected FVS}~\cite{CyganNPPRW11,DBLP:journals/jco/MisraPRSS12}, {\sc Independent FVS}~\cite{DBLP:conf/iwpec/AgrawalGSS16,DBLP:conf/wg/LiP18,DBLP:journals/tcs/MisraPRS12},  {\sc Simultaneous FVS}~\cite{DBLP:journals/toct/AgrawalLMS18,DBLP:journals/corr/Ye15},  {\sc Subset FVS}~
\cite{DBLP:journals/siamdm/CyganPPW13,DBLP:journals/siamcomp/IwataWY16,DBLP:conf/focs/IwataYY18,DBLP:journals/jct/KawarabayashiK12,DBLP:journals/talg/LokshtanovRS18}, {\sc  Pseudoforest Deletion}~\cite{BodlaenderOO18,PhilipRS18}, {\sc Generalized Pseudoforest Deletion}~\cite{PhilipRS18}, and {\sc Almost Forest Deletion}~\cite{RaiS18,LinFWCFL18}. 

%PseudoForest or Pseudoforest

\subsection{Our Problems, Results and Methods}
In this paper we study three problems around \fvs, namely, {\sc Independent FVS}, {\sc Almost Forest Deletion}, and {\sc  Pseudoforest Deletion}. We first define the generalizations of forests that are considered in these problems. We say that a graph $F$ is an $\ell$-forest, if we can delete at most $\ell$ edges from  $F$ to get a forest. That is, $F$ is at most $\ell$ edges away from being a forest. On the other hand, a {\em pseudoforest} is an undirected graph, in which every connected component has at most one cycle. Now, we are ready to define our problems. 

\begin{description}
\setlength{\itemsep}{2pt}
\item[{\sc Independent FVS (IFVS)}:] Given a graph $G$ and a non-negative integer $k$, does there exist a {\sf fvs} $S$ of size at most $k$, that is also an {\em independent set} in $G$?
\item[{\sc Almost Forest Deletion (AFD)}:] Given a graph $G$ and two non-negative integers $k$ and $\ell$, does there exist a vertex subset $S$ of size at most $k$ such that $G-S$ is an $\ell$-forest?
\item[{\sc Pseudoforest Deletion (PDS)}:] Given a graph $G$ and a non-negative integer $k$, does there exist a vertex subset $S$ of size at most $k$ such that $G-S$ is a pseudoforest?
\end{description}

Given an instance of \fvs, by subdividing every edge we get an instance of {\sc Independent FVS}, which is a reduction from \fvs to {\sc Independent FVS} leaving k unchanged showing that it generalizes \fvs.  On the other hand setting $\ell=0$ in {\sc Almost Forest Deletion} results in \fvs. 
The best known algorithms for {\sc Independent FVS}, {\sc Almost Forest Deletion}, and 
{\sc  Pseudoforest Deletion} are $\mathcal{O}^{\star}(3.619^k)$~\cite{DBLP:conf/wg/LiP18}, $\mathcal{O}^{\star}(5^k4^\ell)$~\cite{LinFWCFL18}, and 
$\mathcal{O}^{\star}(3^k)$~\cite{BodlaenderOO18}, respectively. Our main objective is to improve over these running times for the corresponding problems. In a nutshell our paper is as follows.

\begin{tcolorbox}[colback=gray!5!white,colframe=gray!75!black]
Motivated by the methodology developed by Li and Nederlof~\cite{li-soda20} for {\sc FVS}, we relook at several problems around {\sc FVS}, such as {\sc Independent FVS}, {\sc Almost Forest Deletion}, and 
{\sc  Pseudoforest Deletion}, and design the current fastest randomized algorithm for these problems.  Our results show that the method of  Li and Nederlof~\cite{li-soda20} is extremely broad and should be applicable to more problems. 
\end{tcolorbox}

%\todo[inline]{write something here} %written below
To achieve improvements and tackle {\sc Independent FVS} and {\sc Almost Forest Deletion} at once, we propose a more generalized version of the {\sc Almost Forest Deletion} problem.
% as follows: 

\defparprob{\sc Restricted Independent Almost Forest Deletion (RIAFD)}{A graph $G$, a set $R\subseteq V(G)$, and integers $k$ and $\ell$}{$k$ and $\ell$}{Does there exist a set $S\subseteq V(G)$ of size at most $k$ that does not contain any element from $R$, that is also an independent set in $G$, and  $G - S$ is an $\ell$-forest?}
%\begin{description}
%    \item[{\sc Restricted Independent Almost Forest Deletion ($R$-IAFD):}] Given a graph $G$, a set $R$ and integers $k$ and $\ell$, does there exist a set $S$ of size at most $k$ that does not contain any element from $R$, that is also an independent set in $G$ such that $G - S$ is an $\ell$-forest.
%\end{description}

\noindent 
Setting $\ell = 0, R = \varnothing$ we get the {\sc Independent FVS} problem. 
A simple polynomial time reduction, where we subdivide every edge and add all the subdivision vertices to $R$, yields an instance of {\sc RIAFD}, given an instance of {\sc Almost Forest Deletion}.
The reduction leaves $\ell$ and $k$ unchanged. 

% we state later will yield an instance of {\sc $R$-IAFD} given an instance of {\sc Almost Forest Deletion} while leaving $\ell$ and $k$ unchanged. 

To describe our results, we first summarize the method of Li and Nederlof~\cite{li-soda20}(for FVS) which we adopt accordingly. 
%  can be summarized as follows. 
The main observation guiding the method is the fact that after doing some simple preprocessing on the graph, we can ensure that 
a {\em large fraction of edges are incident on  every solution} to the problem. This leads to two-step   algorithms, one for the dense case and the other for the sparse case. In particular, if we are aiming for an algorithm with running time $\cO^\star(\alpha^k)$, then we do as follows. 
\begin{description}
\item[{Dense Case:}] In this case, the number of edges incident to any FVS of size $k$ is superlinear(in $k$), and we select a vertex into our solution with probability at least $\frac{1}{\alpha}$. 
\item[Sparse Case:] Once the dense case is done, we know that we have selected vertices, say $k_1$,  with probability $( \frac{1}{\alpha})^{k_1}$. Now, we know that the number of edges incident to an FVS of the graph is 
$\cO(k)$ and the existence of solution $S$ of size at most $k$, implies that the input graph has treewidth at most $k+1$. Now, using this fact and the fact that deleting the solution leaves a graph of constant treewidth, we can actually show that graph has treewidth $(1-\Omega(1))k=\gamma k$. This implies that if we have an algorithm on graphs of treewidth ($\tw$) with running time $\beta^\tw$, such that $\beta^\gamma\leq \alpha$, then we get the desired algorithm with running time   $\cO^\star(\alpha^k)$.
\end{description}

So a natural approach for our problems which are parameterized by solution size is to devise an algorithm using another algorithm parameterized by treewidth with an appropriate base in the exponent, along with  probabilistic reductions with a good success probability. However, to get the best out of methods of Li and Nederlof~\cite{li-soda20}, it is important to have an algorithm parameterized by treewidth that is based on \emph{Cut \& Count} method~\cite{cutandcount}. However, for all the algorithms for problems we consider, only non  \emph{Cut \& Count} algorithms were known. Thus, our first result is as follows. 

\begin{theorem}\label{thm:cut&count}
    There exists an $\mathcal{O}^\star\left(3^{\tw}\right)$ time Monte-Carlo algorithm that given a tree decomposition of the input graph of width $\tw$ solves the following problems:
    \begin{enumerate}
        \item {\sc Restricted-Independent Almost Forest Deletion} in exponential space.
        \item {\sc Pseudoforest Deletion} in exponential space.
    \end{enumerate}
\end{theorem}

Note that a yes-instance of {\sc RIAFD} has treewidth $k+\ell+1$. Thus as our first result, we design a randomized algorithm based on Theorem~\ref{thm:cut&count} and iterative compression with running time $\mathcal{O}^{\star}(3^k\cdot 3^\ell)$ for {\sc RIAFD}. This yields $\mathcal{O}^{\star}(3^k)$ and $\mathcal{O}^{\star}(3^k \cdot 3^\ell)$ running time algorithms for {\sc Independent FVS} and {\sc Almost Forest Deletion}, respectively, which take polynomial space (though, these do not appear in literature).
Next, we devise probabilistic reduction rules to implement the first step in the method of Li and Nederlof~\cite{li-soda20}. We analyze these rules by modifying the analysis of their lemmas to get an 
$\OO^\star(2.85^k \cdot 8.54^\ell)$ time algorithm that takes polynomial space, and an $\OO^\star(2.7^k\cdot 36.61^\ell)$ time  algorithm that takes exponential space for solving {\sc RIAFD}.  All these algorithms while progressively improving the dependence on $k$ slightly, significantly worsen the dependence on $\ell$. Therefore, to obtain an algorithm with an improved dependence on $\ell$ we describe a procedure to construct a tree decomposition of width $ k + \frac{3}{5.769}\ell +\cO(\log(\ell))$ given a solution of size $k$ to an instance of {\sc RIAFD}. This procedure when combined with an iterative compression routine yields an $\OO^\star(3^k \cdot 1.78^\ell)$ algorithm for {\sc RIAFD}. This brings us to the following result. 

\begin{theorem}\label{thm:afd}
    There exist  Monte-Carlo algorithms that solve {\sc RIAFD} problem in 
    %with high probability in:
    \begin{enumerate}
        \item $\mathcal{O}^{\star}(3^k \cdot 3^\ell)$ time and polynomial space.
        \item $\mathcal{O}^{\star}(2.85^k \cdot 8.54^\ell)$ time and polynomial space.
        \item $\mathcal{O}^{\star}(2.7^k \cdot 36.61^\ell)$ time and exponential space. 
        \item $\mathcal{O}^{\star}(3^k \cdot 1.78^\ell)$ time and exponential space. 
    \end{enumerate}
\end{theorem}

%Pseudoforest

As a corollary to Theorem~\ref{thm:afd}, we get the following result about {\sc Independent FVS}.

\begin{theorem}\label{thm:ifvs}
    There exist  Monte-Carlo algorithms that solve {\sc Independent FVS}  in:
    \begin{enumerate}
        \item $\mathcal{O}^\star(3^{\tw})$ time, given a tree decomposition of width $\tw$.
       % \item $\mathcal{O}^{\star}(3^k)$ time and polynomial space
        \item $\mathcal{O}^{\star}(2.85^k)$ time and polynomial space
        \item $\mathcal{O}^{\star}(2.7^k)$ time and exponential space
    \end{enumerate}
\end{theorem}

Although we have a deterministic $\mathcal{O}^{\star}({3^k})$ algorithm for \textsc{Pseudoforest deletion} given by Bodlaender et al.~\cite{BodlaenderOO18} which runs in exponential space, to make use of the techniques from \cite{li-soda20} we develop our \emph{Cut \& Count} algorithm which has the same asymptotic running time. However, even with our  \emph{Cut \& Count} algorithm, we cannot make full use of the methods of Li and Nederlof~\cite{li-soda20} and only get the following improvement.
% over the previous algorithm. 

%Further improvement using the aforementioned techniques gives us an $\mathcal{O}^{\star}(2.8446)^k$ algorithm. 

\begin{theorem}\label{thm:pseudo}
There exists a Monte-Carlo algorithm that solves {\sc Pseudoforest Deletion} in $\mathcal{O}^{\star}(2.85^k)$ time and polynomial space. 
%problem with high probability in:
%    \begin{enumerate}
%    %    \item in $\mathcal{O}^{\star}(3^k)$ time and polynomial space. 
%        \item in $\mathcal{O}^{\star}(2.8446^k)$ time and polynomial space
%    \end{enumerate}
\end{theorem}

\section{Preliminaries}\label{sec:prelims}
%!TEX root = main.tex

For a set $A$, $\binom{A}{\cdot,\cdot,\cdot}$ denotes the set of all partitions of $A$ into three subsets.

Let $G(V,E)$ or $G=(V,E)$ be an undirected graph, where $V$ is the set of vertices and $E$ is the set of edges. We also denote $V(G)$ to be the vertex set and $E(G)$ to be the edge set of graph $G$. Also, $|V| = n$ and $|E|= m$.
For a vertex subset $S\subseteq V(G)$, $G[S]$ denotes the subgraph induced on the vertex set $S$.
For $S,T \subseteq V$, $E[S,T]$ denotes the edges intersecting both $S$ and $T$. 
%$G[S]$ denotes the graph induced on the vertex set $S$.
For a vertex subset $V'$, the graph $G - V'$ denotes the graph $G[V\setminus V']$. For an edge subset $E'$, the graph $G - E'$ denotes the graph $G'=(V,E\setminus E')$.  For a vertex $v\in V$, $deg(v)$ denotes the degree of the vertex, i.e., the number of edges incident on $v$.  For a vertex subset $S\subseteq V(G)$, $deg(S)=\sum_{v\in S}deg(v)$. 
Given an edge $e=(u,v)$, the subdivision of the edge $e$ is the addition of a new vertex between $u$ and $v$, i.e. the edge $e$ is replaced by two edges $(u,w)$ and $(w,v)$, where $w$ is the newly added vertex. Here, $w$ is called a ``subdivision vertex''. Now, we make note of the following lemma on the number of connected components of a forest.

\begin{lemma}[\cite{cutandcount}]  \label{lem:forestcc}
A graph with $n$ vertices and $m$ edges is a forest iff it has at most $n-m$ connected components.
\end{lemma}

% We will be making extensive use of tree decompositions and \emph{Cut \& Count}. Please refer to \cite{parameterizedcomplexity} and \cite{cutandcount} for background and terminology on treewidth and \emph{Cut \& Count}.

% Shift to appendix %

\begin{definition}
A \emph{tree decomposition} of a graph $G=(V,E)$ is a pair $\mathbb{T}=(\{B_x\;|\;x\in I\},T=(I,F))$ where $T$ a tree and $\{B_x\;|\;x\in I\}$ is a collection of subsets (called bags) of $V$, such that
\begin{enumerate}
    \item $\bigcup_{x\in I}B_x=V$.
    \item For all $(u,v)\in E$ there is an $ x\in I$ with $\{u,v\}\subseteq B_x$.
    \item For all $v\in V$, the set of nodes $\{x\in I\;|\:v\subseteq B_x\}$ forms a connected subtree in $T=(V,I)$.
\end{enumerate}
The \emph{width} of the tree decomposition 
%$(\{B_x\;|\;x\in I\},T=(I,F))$ 
$\mathbb{T}$ is $\text{max}_{x\in I}\:|B_x|-1$. The \emph{treewidth} of a graph $G$, denoted by $\tw(G)$, is the minimum width over all tree decompositions of $G$.
\end{definition}

 We sometimes abuse notation and use $\tw(\mathbb{T})$ to denote the width of the tree decomposition $\mathbb{T}$. For the definition above, if there are parallel edges or self loops we can just ignore them, i.e., a tree decomposition of a graph with parallel edges and self loops is just the tree decomposition of the underlying simple graph (obtained by keeping only one set of parallel edges and removing all self loops).

There is also the notion of a \emph{nice} tree decomposition, which is used in this paper. In literature, there are a few variants of this notion that differ in details. We use the one with \emph{introduce edge} nodes and root bag and leaf bags of size zero. A nice tree decomposition is a tree decomposition $(\{B_x\;|\;x\in I\},T=(I,F))$ where $T$ is rooted tree and the nodes are one of the following five types. With each bag in the tree decomposition, we also associate a subgraph of $G$; the subgraph associated with bag $x$ is denoted $G_x = (V_x, E_x)$. We give each type together with how the corresponding subgraph is formed.

\begin{itemize}
    \item \textbf{Leaf} nodes $x$. $x$ is a leaf of $T$; $|B_x|=0$ and $G_x=(\varnothing,\varnothing)$ is the empty graph.
    \item \textbf{Introduce vertex} nodes $x$. $x$ has one child, say $y$. There is a vertex $v$ with $B_x=B_y\cup \{v\}$, $v\notin B_y$ and $G_x=(V_y\cup \{v\},E_y)$, i.e, $G_x$ is obtained by adding an isolated vertex $v$ to $G_y$.
    \item \textbf{Introduce edge} nodes $x$. $x$ has one child, say $y$. There are two vertices $v,w\in B_x$, $B_x=B_y$ and $G_x=(V_y,E_y\cup \{(v,w)\})$, i.e., $G_x$ is obtained from $G_y$ by adding an edge between these two vertices in $B_x$. If we have parallel edges, we have one introduce edge node for each parallel edge. A self loop with endpoint $v$ is handled in the same way, i.e., there is an introduce edge node with $v\in B_x$ and $G_x$ is obtained from $G_y$ by adding the self loop on $v$.
    \item \textbf{Forget vertex} nodes $x$. $x$ has one child, say $y$. There is a vertex $v$ such that $B_x=B_y \setminus \{v\}$ and $G_x$ and $G_y$ are the same graph.
    \item \textbf{Join} nodes $x$. $x$ has two children, say $y$ and $z$. $B_x=B_y=B_z$, $V_y\cap V_z=B_x$ and $E_y\cap E_z=\varnothing$. $G_x=(V_y\cup V_z,E_y\cup E_z)$, i.e., $G_x$ is the union of $G_y$ and $G_z$, where the vertex set $B_x$ is the intersection of the vertex sets of these two graphs. 
\end{itemize}

For the \emph{Cut \& Count} algorithms, the following lemma is essential. For a family of sets ${\mathcal F}$ over a universe $U$, we say that a weight function $w \colon U \mapsto {\mathbb N}$ isolates ${\mathcal F}$, if there is a unique set $S$ in ${\mathcal F}$ with minimum weight $w(S)$. Here, $w(S)=\sum_{x\in S} w(x)$.

\begin{lemma}\label{lem:isolation}
     {\rm (Isolation Lemma, \cite{Mulmuley1987})} Let $\mathcal{F} \subseteq 2^U$ be a non-empty set family over a universe $U$. For each $u\in U$, choose a weight $\omega \in \{1, 2, \hdots W\}$ uniformly and independently at random. Then $\Pr[\omega \text{ isolates } \mathcal{F}] \ge 1 - |U|/W$.
\end{lemma}

In the \emph{Cut \& Count} algorithms and proofs, for a function $f: S \rightarrow T$, given a set $R$, $f|_R$ refers to the function $f$ with its domain restricted to $R$. Formally, $f|_R$ is a function from $R$ to a subset of $T$ such that $f|_R(r) = f(r)$ for all $r \in R$. Given values $u$ and $v$, $f[u \rightarrow v]$ refers to a function with $u$ in domain and $v$ in range with all mappings from $S$ to $T$ preserved and $u$ mapped to $v$. Formally, $f[u \rightarrow v]$ is a function from $S \cup \{u\}$ to $T \cup \{v\}$ such that $f[u \rightarrow v](s) = f(s)$ for all $s \in S$ and $f[u \rightarrow v](u) = v$. Also, we define $f^{-1}(s) := \{x | x \in S \land f(x) = s\}$. We use the Iverson's bracket notation $[b]$ for a Boolean predicate $[b]$ which denotes $1$ if $b$ is True and $0$ otherwise.

In this paper, we will be dealing with randomized algorithms with one-sided error-probability, i.e. only \emph{false negatives} are possible. The \emph{success-probability} of an algorithm is the probability that the algorithm finds a solution, given that at least one such solution exists. We define \emph{high-probability} to be probability at least $1 - \frac{1}{2^{c|x|}}$ or sometimes $1 - \frac{1}{|x|^{c}}$, where $\vert x\vert$ is the input size and $c$ is a constant. 
Given an algorithm with constant success-probability, we can boost it to high-probability by performing $\mathcal{O}^{\star}(1)$ independent trials. We cite the following folklore observation:
\begin{lemma}\label{lem:randomalgo}{\rm(Folklore, \cite{li-soda20}).} If a problem can be solved with success probability $\frac{1}{S}$ and in expected time $T$, and its solutions can be verified for correctness in polynomial time, then it can be also solved in $\mathcal{O}^{\star}(S \cdot T )$ time with high probability.
\end{lemma}

We will use the following notion of separations in a graph from \cite{li-soda20}:
\begin{definition}{\rm (Simple Separator, \cite{li-soda20}).} Given a graph $G(V,E)$, a partition $(A,B,S) \in \binom{V(G)}{\cdot,\cdot,\cdot}$ of $V$ is a separation if there are no edges between $A$ and $B$.
\end{definition}

\begin{definition}{\rm (Three-Way Separator, \cite{li-soda20}).} Given a graph $G = (V, E)$, a three-way separator is a partition $(S_{\{1\}}, S_{\{2\}}, S_{\{3\}}, S_{\{1,2\}}, S_{\{1,3\}},$ $S_{\{2,3\}}, S_{\{1,2,3\}})$ of $V$ such that there are no edges between any two sets $S_I$, $S_J$ whose sets $I$ and $J$ are disjoint.
\end{definition}

A $\beta$-separator for a graph $G(V,E)$ is a set of vertices whose removal from $G$ leaves no connected component of size larger than $\frac{|V|}{\beta}$, where $\beta > 0$ is some constant. Thus, a $\beta$-separator is a balanced separator of the graph. More generally, one can define a $\beta$-separator with respect to a weight function on the vertices. We now give a method to construct a $\beta$-separator of a graph $G$ given a tree decomposition (Lemma~\ref{lem:generalized_beta_separator}). 
% Towards that we use the following lemma from \cite{li-soda20} for constructing $\beta$-separator of a forest.

% For a given vertex weight function $w:V\rightarrow \mathbb{R}_{\geq 0}$ we define $\beta$-separator of a graph $G$ to be a set of vertices $S$ of the graph such that every connected component of $G-S$ has weight at most a $1/\beta$ times the total weight $w(V):=\sum\limits_{v\in V}w(v)$. We can see that $\beta$-separator is a balanced separator of $G$. 

% \begin{lemma}[\cite{li-soda20}]\label{forsep}
% Given a forest $T=(V,E)$,  $\beta >0$, and a weight function $w:V\rightarrow \mathbb{R}_{\geq 0}$, we can delete a set $S$ of $\beta$ vertices in polynomial time so that every connected component of $T-S$ has  weight at most  $\frac{w(V)}{\beta}$.
% \end{lemma}

\begin{lemma}\label{lem:generalized_beta_separator}
    Given a graph $G(V,E)$ on $n$ vertices with vertex weights $\omega(v)$ and its tree decomposition $\mathbb{T}$ of width $\tw$, for any $\beta > 0$, we can delete a set $S$ of $\beta (\tw+1)$ vertices so that every connected component of $G - S$ has  weight at most $\frac{\omega(V)}{\beta}$ in polynomial time.
\end{lemma}
\begin{proof}
    Given a bag $x$ of the $\mathbb{T}$, we define the weight of the subtree rooted at $x$ ($w(x)$) to be the sum of weights of vertices present in the set formed by union of all bags in the subtree of x. Formally, $w(x) := \sum\limits_{v \in V_x} \omega(v)$. Start with an empty set $S$.
    
    % Now, use Lemma~\ref{forsep} with a slight modification. When deleting a bag $x$, delete the vertices in $B_x$ from all other bags. Lemma~\ref{forsep} will delete at most $\beta$ bags, therefore $S$ will contain at most $\beta(\tw+1)$ vertices. 

    Exhaustively, select a bag $x$ of maximal depth such that $w(x) > \frac{\omega(V)}{\beta}$, then remove the bag $x$ and its subtree and add all vertices in $B_x$ to the set $S$. Also, delete the vertices in $B_x$ from all other bags. Note that the maximality condition assures us that the subtrees rooted at the children of $x$ have total weight at most $\frac{\omega(V)}{\beta}$ each. Moreover, by deleting the subtree rooted at $x$, we remove at least $\frac{\omega(V)}{\beta}$ weight, which can happen at most $\beta$ times. As each bag has size at most $t+1$, the total vertices we select are at most $\beta (t+1)$ to be added to $S$.
    
    To see how there are no connected components of size more than $\frac{\omega(V)}{\beta}$ in $G - S$ left, suppose that the tree decomposition left after following this algorithm is $\mathbb{T'}$. Now assume that a connected component $C$ of weight more than $\frac{\omega(V)}{\beta}$ exists in $G - S$. Then all of its vertices in their entirety must lie inside $\mathbb{T}'$ (since all children of a deleted bag have weight at most $\omega(V)/\beta$). Now, take the vertex of $C$ which is in the least depth in $\mathbb{T}'$ and say it lies in the bag $c$. All the members of its connected components therefore have to appear in the subtree rooted at $c$. Therefore, $w(c) \ge \frac{\omega(v)}{\beta}$ which would mean that this is not the terminal condition for our algorithm. 
    
    From the description of the algorithm it is easy to see that it runs in polynomial time. 
\end{proof}

In \cite{li-soda20}, the authors presented a method involving randomized reductions and small separators to get faster randomized algorithms for \fvs. It turns out that this method can be generalized to work for a certain set of ``vertex-deletion problems''. We will now describe the basic structure of this method and will follow this outline wherever this method is used in the rest of the paper.

Throughout this outline, assume that we are working on some vertex-deletion problem $\mathcal{P}$. Let $G(V,E)$ be the graph involved in a given instance of $\mathcal{P}$. A valid solution $S\subseteq V$ is a set of vertices of $G$ which solves the given problem instance of $\mathcal{P}$. 

The method is divided into two cases: A dense case and a sparse case.

\medskip
\noindent
{\bf Dense Case.}
The algorithm goes into this case when for a given problem all the existing solution sets are of high average degree. In formal terms, every set $S \subseteq V$ of size $k$ which is a valid solution of the given instance satisfies $deg(S) > c\cdot k$, where $c = \Theta(1)$.

To handle this case, a vertex $v\in V$ is sampled randomly based on a weight function $\omega(v)$ which depends on $deg(v)$, deletes $v$ and makes appropriate updates to the parameters. In this paper, we use $\omega(v) = deg(v) - 2$ for all the problems discussed. This process acts like a probabilistic reduction rule for the problem as it may fail with certain probability. 

\medskip
\noindent
{\bf Sparse Case.}
The algorithm goes into this case when for a given problem there exists a solution set which has low average degree. In formal terms, there exists a vertex subset $ S \subseteq V$ of size $k$ which is a valid solution of the given instance and satisfies $deg(S) \le c \cdot k$, where $c = \mathcal{O}(1)$. Due to this reason, the number of edges in the given graph can be bounded, thus the input graph $G$ is sparse.

The proof for the small separator lemma in \cite{li-soda20} doesn't require the remaining graph, i.e. the graph obtained by deleting the solution set, to be a forest only. As long as there is a good $\beta$-separator of the graph $G-S$, the proof works. Lemma \ref{lem:generalized_beta_separator} helps to construct such a $\beta$-separator of size $\beta (\tw+1)$ for a graph with given tree decomposition of width $\tw$.
%Hence, this notion can be generalized.

% \todo[inline]{mention what is $\beta$. Does it depends on $c$?}

The small separator helps to construct a  tree decomposition of small width, given a solution set with bounded degree. The idea suggested in \cite{li-soda20} was to use iterative compression techniques to construct a solution utilizing the small separator. This also requires solving a bounded degree version of the problem, which can be done using \emph{Cut \& Count} based algorithms. Specific details for each problem will be explained in the corresponding sections in due course.

%\medskip
%\noindent
%{\bf Combined Randomized Algorithm:}
%Firstly, we apply basic reduction rules on the given problem instance. We don't know whether the current reduced instance belongs to the dense case or the sparse Case. 
%So, we randomly decide this by a coin toss with suitable weights. \todo{Do you need weight here. Is choosing it with prob 1/2 fine?} 
%%Details will be explained for specific problems in due course.
%%Note that the above method need not apply to all vertex-deletion problems. 
%In this paper, we show that the method works for the Almost Forest Deletion problem and the PseudoForest Deletion problem.
%to be shifted to appendix

\section{Restricted-Independent Almost Forest Deletion}\label{sec:iafd}
%!TEX root = main.tex

In this section we give our algorithm for {\sc RIAFD} and prove Theorem~\ref{thm:afd} and the first part of Theorem~\ref{thm:cut&count}. We first formally show that  {\sc RIAFD} is a generalization of {\sc Almost Forest Deletion}. For any instance $G$ of {\sc Almost Forest Deletion}, subdivide the edges of $G$ and add all the newly created subdivision vertices to $R$. The parameters $k$ and $\ell$ remain the same.

\begin{lemma}\label{lem:afd_iafd}
    Given an instance of {\sc Almost Forest Deletion} $(G(V,E),k,\ell)$, an equivalent instance of 
    {\sc RIAFD}, $(G'(V',E'),k',\ell',R)$, can be constructed as follows:
    \begin{enumerate}
        \item $V' = V, E' = \varnothing, R = \varnothing$.
        \item For each $e=(u,v) \in E$, add a vertex $v_e$ to $V'$ as well as to $R$. Add edge $(u,v_e)$ and $(v_e,v)$ to $E'$ (Essentially, subdivide $e$).
        \item $k' = k, \ell' = \ell$.
    \end{enumerate}
\end{lemma}

In this section, we present fast randomized algorithms for \textsc{RIAFD}. In Section \ref{sec:afd_cutcount} we present a $\mathcal{O}^{\star}(3^{\tw})$ running time algorithm based on the \emph{Cut \& Count} paradigm. Based on this, we give a $\mathcal{O}^{\star}(3^k3^\ell)$ time and polynomial space algorithm in Section \ref{sec:afd_3k3l}. In Section \ref{sec:afd_improve_k}, we further improve the dependence on $k$ by using modified techniques from \cite{li-soda20} to get an algorithm with running time $\mathcal{O}^{\star}(2.85^k8.54^\ell)$ and polynomial space as well as an algorithm with running time $\mathcal{O}^{\star}(2.7^k36.61^\ell)$ but exponential space. Finally, in Section \ref{sec:afd_improve_l}, we improve the dependence on $\ell$ by creating a tree decomposition of width $k + \frac{3}{5.769}\ell$ and run the \emph{Cut \& Count} algorithm presented in Section \ref{sec:afd_cutcount} to get an algorithm with running time $\mathcal{O}^{\star}(3^k1.78^\ell)$. Henceforth, the term {\sf riafd-set} corresponds to a solution for given instance of \textsc{RIAFD} and the term {\sf afd-set} corresponds to a solution for given instance of \textsc{AFD}.
% \todo[inline]{change the {\sc RIAFD} Set to  {\sf riafd-set}}
%Henceforth, the abbreviation {\sc RIAFD} means \textsc{Restricted-Independent Almost Forest %Deletion}, and {\sc RIAFD} Set corresponds to a solution set.
\subsection{$3^{\tw}$ Algorithm} \label{sec:afd_cutcount}

We use the \emph{Cut \& Count} technique to solve {\sc RIAFD} in $\mathcal{O}^*(3^{\tw})$ time. First of all, we require the following lemma.
\begin{lemma}
\label{lem:afcc}
A graph $G = (V, E)$ with $n$ vertices and $m$ edges and non-negative integer $\ell$ is an $\ell$-Forest if and only if it has at most $n-m+\ell$ connected components. 
\end{lemma}
\begin{proof} 
    \textbf{Forward Direction:} By definition of $\ell$-Forest, if you are given an $\ell$-Forest with $n$ vertices and $m$ edges, there exists a set $S$ of $\ell$ edges whose removal leaves a forest with $n$ vertices and $m - \ell$ edges. By Lemma \ref{lem:forestcc}, this leftover forest has at most $n - (m - \ell)$ connected components. Adding back the edges from the set $S$ to the $\ell$-forest cannot result in an increase in the number of connected components. Therefore, the $\ell$-Forest also has at most $n - m + \ell$ connected components.
    
    \textbf{Reverse Direction:} We are given a graph $G = (V, E)$ with $n$ vertices, $m$ edges and at most $n - m + \ell$ connected components. Let the $r$ connected components be $C_1, C_2, \hdots C_{r}$ having $n_1, n_2, \hdots n_{r}$ vertices each respectively. The subgraph consisting of only the spanning trees of the connected components is a forest with $\sum\limits_{i =0}^{r} (n_i - 1) = n - r \ge n - (n- m + \ell) \ge m - \ell$ edges. Let the edge set of the subgraph be $E'$. Therefore, $E \setminus E'$ of cardinality at most $\ell$ is the set of edges to be removed from $G$ to obtain a forest. Therefore, $G$ is an $\ell$-Forest.
\end{proof}
    
Moving on to the \emph{Cut \& Count} Algorithm. Firstly, we define the set $U = V$. We assume that we are given a weight function $\omega : U \rightarrow \{1,\hdots, N\}$, where $N$ is some fixed integer.
    
 \noindent 
 \textbf{The Cut Part:} For integers $A, B, W$ we define:
    \begin{enumerate}
        \item $\mathcal{R}_W^{A, B}$ to be the family of solution candidates: $\mathcal{R}_W^{A, B}$ is the family of sets $X$, where $X \subseteq V$,  $|X| = A$, $G[X]$ contains exactly $B$ edges, $(V \setminus X) \cap R = \varnothing$, $G[V \setminus X]$ is an independent set and $\omega(V \setminus X) = W$;
        \item $S_W^{A, B}$ to be the set of solutions: the family of sets $X$, where $X \in \mathcal{R}_W^{A, B}$ and $G[X]$ is an $\ell$-Forest;
        \item $\mathcal{C}_W^{A, B}$ to be the family of pairs $\big(X, (X_L,X_R)\big)$, where $X \in \mathcal{R}_W^{A, B}$ and $(X_L, X_R)$ is a consistent cut of $G[X]$.
    \end{enumerate}
    Observe that the graph $G$ admits an Restricted Independent Almost Forest Deletion set $F \subseteq V$ of size $k$ if and only if there exist integers $B, W$ such that the set $S_W^{n-k, B}$ is non-empty.
    
\noindent     
\textbf{The Count Part:} Note that for any $A, B, W, X \in \mathcal{R}_W^{A, B}$, there are $2^{cc(G[X])}$ cuts $(X_L, X_R)$ such that $\big(X, (X_L, X_R)\big) \in \mathcal{C}_W^{A, B}$, where by $cc(G[X])$ we denote the number of connected components of $G[X]$.
    
    Now we describe a procedure that, given a nice tree decomposition $\mathbb{T}$, weight function $\omega$ and integers $A, B, W, t$, computes $|\mathcal{C}_W^{A, B}|$ modulo $2^t$ using dynamic programming.
    
    For every bag $x \in \mathbb{T}$, integers $0 \le a \le |V|$, $0 \le b < |V|$ $0 \le w \le \omega(V)$ and $s \in \{\text{\textbf{F}},\text{\textbf{L}}, \text{\textbf{R}}\}^{B_x}$ (called the coloring), define:
    
    \begin{eqnarray*}
        \mathcal{R}_x(a,b,w) & = & \Big{\{} X \;\big| \;X\subseteq V_x \wedge |X| = a \wedge |E_x\cap E(G[X])|=b \wedge \; \left(V_x \setminus X\right) \cap R = \varnothing \wedge\\ & & |E(G[V_x \setminus X])| = 0 \wedge \; \omega(V_x \setminus X) = W \Big{\}}\\
        \mathcal{C}_x(a,b,w) & = & \Big{\{} \big(X, (X_L,X_R)\big)\; \big| \; X \in \mathcal{R}_x(a,b,w) \wedge \\ & & \big(X, (X_L, X_R)\big) \text{ is a consistently cut subgraph of }G_x \Big{\}}\\
        A_x(a,b,w,s) & = & \Big{|}\Big{\{} \big(X, (X_L,X_R)\big) \in \mathcal{C}_x(a,b,w) \; \big| \; \\ & & \big(s(v) \in \{\text{\textbf{L}},\text{\textbf{R}}\} \implies v \in X_{s(v)} \big) \wedge \big(s(v)=\text{\textbf{F}} \implies v \notin X \big) \Big{\}}\Big{|}
    \end{eqnarray*}
    
    The algorithm computes $A_x(a,b,w,s)$ for all bags $x \in \mathbb{T}$ in a bottom-up fashion for all reasonable values of $a$, $b$, $w$ and $s$. We now define the recurrence for $A_x(a,b,w,s)$ that is used by the dynamic programming algorithm. Let $v$ denote the vertex introduced and contained in an introduce vertex bag, $(u,v)$ the edge introduced in the introduce edge bag, and let $y$, $z$ stand for the left and right child of $x \in \mathbb{T}$. Assume all computations to be modulo $2^t$.
    \begin{itemize}
        \item \textbf{Leaf bag:} 
        \begin{align*}
            A_x(0,0,0,\varnothing) &= 1
        \end{align*}
        \item \textbf{Introduce vertex bag:}
        \begin{align*}
            A_x(a,b,w,s \cup \{(v,\text{\textbf{F}})\}) &= [v \notin R] \; A_y(a,b,w - \omega(v),s) \\
            A_x(a,b,w,s \cup \{(v,\text{\textbf{L}})\}) &= A_y(a-1, b, w, s)\\
            A_x(a,b,w,s \cup \{(v,\text{\textbf{R}})\}) &= A_y(a-1, b, w, s)
        \end{align*}
        \item \textbf{Introduce edge bag:}
        \begin{align*}
            A_x(a,b,w,s) &= [s(u)\neq s(v) \iff (s(u)=\text{\textbf{F}} \lor s(v)=\text{\textbf{F}})]A_y(a,b-[s(u) = s(v) \neq \text{\textbf{F}}], w, s)
        \end{align*}
        \item \textbf{Forget bag:}
        \begin{align*}
            A_x(a,b,c,w,s) &= \sum\limits_{\alpha \in \{\text{\textbf{F}}, \text{\textbf{L}}, \text{\textbf{R}}\}} A_x(a,b,w,s[v\rightarrow \alpha])
        \end{align*}
        \item \textbf{Join bag:}
        \begin{align*}
            A_x(a,b,w,s) &= \sum\limits_{\mathclap{\substack{a_1 + a_2 = a + |s^{-1}(\{\text{\textbf{L}},\text{\textbf{R}}\})|\\
            b_1 + b_2 = b \\
            w_1 + w_2 = w + \omega(s^{-1}(\{\text{\textbf{F}}\}))
            }}}A_y(a_1, b_1, w_1, s)\cdot A_z(a_2, b_2, w_2, s)
        \end{align*}
    \end{itemize}

Let $r \in \mathbb{T}$ be the root bag. Therefore, $A_r(A,B,W,\varnothing) \equiv |\mathcal{C}_W^{A, B}| \ (\textrm{mod} \ 2^t)$ which is our required answer.

\begin{lemma}
    \label{lem:afdccmain}
    Let $G(V,E)$ be a graph and $d$ be an integer. Pick $\omega'(v) \in \{1, \ldots, 2 |V|\}$ uniformly and independent at random for every $v \in V$, and define $\omega(v):= |V|^2\omega'(v) + deg(v)$ and $n = |V|$. The following statements hold:
    \begin{enumerate}
        \item If for some integers $m'$, $W = i|V|^2 + d$ we have that $|\mathcal{C}_W^{n-k,m'}| \not\equiv 0\ (\textrm{mod} \ 2^{n-k-m'+l + 1})$, then $G$ has a {\sf riafd-set} $F$ of size $k$ satisfying $deg(F) = d$.
        
        \item If $G$ has a {\sf riafd-set} $F$ of size $k$ satisfying $deg(F) = d$, then with probability at least $1/2$ for some $m'$, $W = i|V|^2 + d$ we have that $|\mathcal{C}_W^{n-k,m'}| \not\equiv 0\ (\textrm{mod}  \ 2^{n-k-m'+\ell + 1})$.
    \end{enumerate}
\end{lemma}
\begin{proof}
This proof is similar to the one for {\sc fvs} in \cite{li-soda20}.

\textbf{Item 1:} Note that if $|\mathcal{C}_W^{n-k,m'}| \not\equiv 0\ (\textrm{mod} \ 2^{n-k-m'+\ell + 1})$, then there must be some vertex subset $F$ of size $k$ such that $F \cap R = \varnothing$, $G[F]$ is an independent set and the number of choices of $X_L$, $X_R$ with $(V \setminus F, (X_L, X_R)) \in \mathcal{C}_W^{n-k,m'}$ is not a multiple of $2^{n-k-m'+\ell + 1}$. Due to independence in choice of cuts for connected components of $G[V \setminus F]$ on whether to put it in $X_L$ or $X_R$ $G[V \setminus F]$ must have at most $n-k-m'+\ell$ connected components. Therefore, by Lemma \ref{lem:afcc} $G[V \setminus F]$ must be an $\ell$-Forest, making $F$ a {\sf riafd-set} of size $k$. The condition on degree follows from the weighting.

\textbf{Item 2:} First apply Lemma \ref{lem:isolation} with $U = V$ and the set family $\mathcal{F}$ being the set of all {\sf riafd-set} $F$ satisfying $deg(F) = d$ with weighting done based on $\omega'$. With probability $1/2$, there will be some weight $i$ such that there is a unique {\sf riafd-set} $F$ with $def(F) = d$ and weight $i$. Therefore, for the weight function $\omega$, we have $W = \omega(F) = i|V|^2 + d$. Since $\omega'$ isolated $F$ out of $\mathcal{F}$ and $d < |V|^2$ (for $k >0$), this is the only $F$ which has a contribution in $\mathcal{C}_W^{n-k,m'}$ that is not a multiple of $2^{n-k-m'+\ell + 1}$ as it has $2^{cc(G[V \setminus F])} \le 2^{n-k-m'+\ell}$ valid cuts.
\end{proof}

While it is clear from the DP and Lemma \ref{lem:afdccmain} that we can get an $\mathcal{O}^\star\left(3^{\tw}\right)$ running time, we will provide the details of a slightly more generalized algorithm which is able to utilize additional structure in the tree decomposition and improves the space bound.

\subsection{$3^{k+\ell}$ Algorithm in Polynomial Space}\label{sec:afd_3k3l}
The above \emph{Cut \& Count} algorithm  utilizes exponential space. We notice that in all the problems discussed in the paper, the tree decomposition that we have always has a {\em  large set which is present in all bags of the tree decomposition}. We will exploit this structure to obtain a polynomial space algorithm.

\begin{definition}
    Given a set $S \subseteq V$ and a function $f: S \leftarrow \left\{\text{\textbf{F}},\text{\textbf{L}},\text{\textbf{R}}\right\}$ we define the quantity $\mathcal{C}_{W,f}^{A, B}$ as follows:
    
    \begin{eqnarray*}
    \mathcal{C}_{W,f}^{A, B} & = & \Big\{\big(X, (X_L,X_R)\big)\, \Big|\: X \in \mathcal{R}_W^{A, B} \land (X_L, X_R) \text{ is a consistent cut of } G\big[X\big] \land \\ & & \big(v \in S \implies v \text{ agrees with } f\big) \Big\}.
    \end{eqnarray*}
    
    where ``$v$ agrees with $f$'' means that $v \in V \setminus X$ if $f(v) = \text{\textbf{F}}$, $v \in X_L$ if $f(v) = \text{\textbf{L}}$ and $v \in X_R$ if $f(v) = \text{\textbf{R}}$.
\end{definition}

\begin{claim}
\label{claim:afdpoly}
     Given a tree decomposition $\mathbb{T}$ with a set $S\subseteq V$ which is present in all its bags, a fixed integer $t$ and a function $f: S \rightarrow \left\{\text{\textbf{F}},\text{\textbf{L}},\text{\textbf{R}}\right\}$, there is a routine \texttt{RIAFD-FCCount}($\mathbb{T},R,A,B,W,f,t$) which can compute $|\mathcal{C}_{W,f}^{A, B}| \;(\textrm{mod} \; 2^t)$ in time $\mathcal{O}^\star\left(3^{\tw - |S|}\right)$. 
\end{claim}
\begin{claimproof}
    We will give a brief description of the routine \texttt{RIAFD-FCCount} as that will suffice to prove this claim.
    In every entry of the DP table described for $|\mathcal{C}_W^{A, B}|$, just compute all values of $A_x(a,b,w,s)$, where $s|_{B_x \cap S} = f|_{B_x \cap S}$ and ignore all computations that do not agree to this condition. This means per bag, only $\mathcal{O}^\star(3^{\tw - |S|})$ computations are required (since in all bags only at most $\tw+1 - |S|$ values of $s$ are not ``fixed'' by $f$). The required answer is in the root bag $r$ as the entry $A_r(A,B,W,\varnothing) \equiv |\mathcal{C}_{W,f}^{A, B}| \ (\textrm{mod} \ 2^t)$.
\end{claimproof}

Now, given a tree decomposition $\mathbb{T}$ with a set $S\subseteq V$ which is present in all its bags, we can see that,
$$|C_W^{A,B}| = \sum\limits_{\text{All possible } f: S \rightarrow \left\{\text{\textbf{F}},\text{\textbf{L}},\text{\textbf{R}}\right\}} |\mathcal{C}_{W,f}^{A, B}|.$$

Now, we define a procedure \texttt{RIAFDCutandCount} which given a tree decomposition $\mathbb{T}$, a set $S\subseteq V$ present in all bags of $\mathbb{T}$ uses the above fact to improve the space bound from $\mathcal{O}^{\star}\left(3^{\tw}\right)$ to $\mathcal{O}^{\star}\left(3^{\tw - |S|}\right)$.

\begin{algorithm}[!ht]
\label{alg:afd_cc}
\DontPrintSemicolon
\caption{$\mathtt{RIAFDCutandCount}$$(\mathbb{T},R,k,\ell,S)$}
\SetKwInOut{Input}{Input}\SetKwInOut{Output}{Output}
\SetKwData{afdfcc}{RIAFD-FCCount}
\SetKwData{infs}{Infeasible}
\SetKwData{c}{count}
\SetKwData{flag}{flag}
\SetKw{Break}{break}
\SetKw{Return}{return}
\Input{Tree decomposition $\mathbb{T}$, $G = (V,E)$, set $R$, parameters $k \le n$ and $\ell = \mathcal{O}(n^2)$}
\Output{A {\sf riafd-set} $F$ of size at most $k$ or \infs}
\Begin{
    \For{$n-k \le A \le n$, $0 \le B \le A +\ell-1$, $0 \le W \le 2|V|^4 + 2|E|$}{
        $t \leftarrow A - B + \ell + 1$\;
        \For{$n^{\mathcal{O}(1)}$ iterations}{
            \c $\leftarrow 0$\;
            Randomly initialize $\omega$ as stated in Lemma \ref{lem:afdccmain} considering $\mathbb{G} = G$\;
            \For{all possible $f: S \rightarrow \left\{\text{\textbf{F}},\text{\textbf{L}},\text{\textbf{R}}\right\}$}{
                \c $\leftarrow$ \c $+$ $\mathtt{RIAFD-FCCount}$($\mathbb{T}$,$R$,$A$,$B$,$W$,$f$,$t$)\;
            }
            \If{\c $\not\equiv 0 \ (\textrm{mod}\ 2^t)$}{
                $F \leftarrow$ a {\sf riafd-set} of $G$ constructed using self-reduction\;
                \Return $F$
            }
        }
    }
\Return \infs\;
}
\end{algorithm}

\begin{theorem}
\label{thm:afdcc}
  Given a tree decomposition $\mathbb{T}$, a set $S\subseteq V$ present in all bags of $\mathbb{T}$, a set $R$ and parameters $k$ and $\ell$, $\mathtt{RIAFDCutandCount}$ solves {\sc RIAFD} in $\mathcal{O}^{\star}\left(3^{\tw}\right)$ time and $\mathcal{O}^{\star}\left(3^{\tw - |S|}\right)$ space with high probability.
\end{theorem}
\begin{proof}
    We first prove the probability bound. By Lemma \ref{lem:afdccmain} Item $(2)$ if a {\sf riafd-set} of size at most $k$ exists, then for some values satisfying $n-k \le A \le |V|$, $0 \le B \le A +\ell+1$ and $0 \le W \le 2|V|^4 + 2|E|$, in each iteration of the for block starting at Line $5$ \texttt{count} $\not\equiv 0 \ (\textrm{mod}\ 2^t)$ with probability $1/2$. Lemma \ref{lem:afdccmain} Item $(1)$ makes it so that whenever we have \texttt{count} $\not\equiv 0 \ (\textrm{mod}\ 2^t)$, there is guaranteed to be a {\sf riafd-set} i.e. there are no false positives. Therefore, in $n^{\cO(1)}$ iterations, we obtain the required {\sf riafd-set}, if it exist, with high probability and if such a set doesn't exist, \texttt{RIAFDCutandCount} will always return \texttt{Infeasible}.
    
    Now, to prove the time and space complexity bounds, we first take note of the fact that by Claim \ref{claim:afdpoly} Line $8$ takes $\mathcal{O}^{\star}\left(3^{\tw - |S|}\right)$ time and $\mathcal{O}^{\star}\left(3^{\tw - |S|}\right)$ space. Since the number of possible $f: S \rightarrow \left\{\text{\textbf{F}},\text{\textbf{L}},\text{\textbf{R}}\right\}$ is $3^{|S|}$, Line $8$ runs for $\mathcal{O}^{\star}\left(3^{\tw}\right)$ time but since each run is independent it still requires only $\mathcal{O}^{\star}\left(3^{\tw - |S|}\right)$ space. All other lines contribute at most polynomial cost overall to the total running time and space. Therefore, the time and space bounds for Line $8$ are the ones for the complete algorithm. 
\end{proof}

\begin{lemma}
\label{lem:afd_kltw}
Given a graph $G(V,E)$ and a {\sf riafd-set} $F$ of size $k$ you can construct a tree-decomposition $\mathbb{T}$ which contains a set $S \supseteq F$ of size at most $k+\ell$ in all bags and has width at most $|S|+1$ in polynomial time.
\end{lemma}
\begin{proof}
Initially the set $S = F$. $G[V \setminus F]$ is an $\ell$-forest. Now, find any spanning tree of each connected component. We can see that the union of the spanning trees is the forest with maximum number of edges that spans $G[V \setminus F]$. Therefore, there can be at most $\ell$ edges that were left out from the forest since $G[V \setminus F]$ is an $\ell$-forest. Add one end-point from each of these leftover edges to the set $S$.
This set $S$ is now an {\sf fvs} of $G$ of size at most $k+\ell$. Therefore, we can construct a tree decomposition $\mathbb{T}$ of width $1$ of the forest $G[V \setminus S]$. Add the set $S$ to all bags of $\mathbb{T}$. Therefore, width of $\mathbb{T}$ is now at most $|S| + 1$. This completes our construction. It's easy to see from the description of the construction procedure that it takes polynomial time. 
\end{proof}

\begin{algorithm}[!ht]
\label{alg:afd_3k3l}
\DontPrintSemicolon
\caption{$\mathtt{RIAFD3k3l}$$(G,R,k,\ell)$}
\SetKwInOut{Input}{Input}\SetKwInOut{Output}{Output}
\SetKwData{afdcc}{RIAFDCutandCount}
\SetKwData{infs}{Infeasible}
\SetKwData{c}{count}
\SetKwData{flag}{flag}
\SetKw{Break}{break}
\SetKw{Return}{return}
\Input{Graph $G = (V,E)$, set $R$, parameters $k \le n$ and $\ell = \mathcal{O}(n^2)$.}
\Output{A {\sf riafd-set} $F$ of size at most $k$ or \infs.}
\Begin{
Order the vertices $V$ arbitrarily as $(v_1, v_2, \ldots, v_n)$ \;
$F \leftarrow \varnothing$\;
\For{$i = 1, 2, \ldots, n$}{
    $\mathbb{T} \leftarrow$ Compute the tree decomposition of $G\left[\left\{v_1, \ldots, v_{i-1}\right\}\right]$ by Lemma \ref{lem:afd_kltw} on input $F$\;
    $S\leftarrow$ F\;
    Add $v_i$ to all bags of $\mathbb{T}$ and to $S$\;
    $F \leftarrow$ $\mathtt{RIAFDCutandCount}$($\mathbb{T}$,$R$,$k$,$\ell$,$S$)\;
    \If{$F$ is \infs}{
        \Return \infs\;
    }
}
\Return $F$\;
}
\end{algorithm}

Now, we restate Theorem \ref{thm:afd} (1) and prove it. Also, note that we set $R=\varnothing$ in all the cases except for when there is an explicit requirement of a restricted set of vertices.

$\blacktriangleright$ \textbf{Theorem \ref{thm:afd} (1)}.
{\em The randomized algorithm \hyperref[alg:afd_3k3l]{$\mathtt{RIAFD3k3l}$} solves {\sc Almost Forest Deletion} in $\mathcal{O}^{\star}\left(3^{k}3^{\ell}\right)$ time and polynomial space with high probability.}
\begin{proof}
Suppose that there exists a {\sf riafd-set} $F^\star$ of size at most $k$. Let $\left(v_1, \ldots, v_i\right)$ be the ordering from Line $2$, and define $V_i := \left\{v_1, \ldots, v_i\right\}$. Observe that $F^\star \cap V_i$ is a {\sf riafd-set} of $G[V_i]$, so {\sc RIAFD} problem on Line $8$ is feasible. Line $8$ correctly computes a {\sf riafd-set} with high probability on any given iteration. Therefore, with high probability, such a {\sf riafd-set} for $G$ is returned by a union bound. 

We now bound the running time and space complexity. On Line $5$, the current set $F$ is a {\sf riafd-set} of $G[V_{i-1}]$, so Lemma \ref{lem:afd_kltw} guarantees tree decomposition $\mathbb{T}$ of width at most $k+\ell+1$, and adding $v_i$ to each bag on Line $7$ increases the width by at most $1$. Also Lemma \ref{lem:afd_kltw} guarantees a set $S \le k+\ell$ such that $\tw(\mathbb{T}) - |S| \le 1$ and adding $v_i$ to the set $S$ increases its size by $1$. Therefore, by Theorem \ref{thm:afdcc}, Line $8$ runs in time $\mathcal{O}^\star(3^{k+\ell})$ and space $\mathcal{O}^\star(1)$ as desired.  
\end{proof}

\subsection{Improving the Dependence on $k$}\label{sec:afd_improve_k}
In this subsection, we try to reduce the dependence on $k$ by allowing an increase in dependence on $\ell$. We use the method of \cite{li-soda20} using the outline described in Section \ref{sec:prelims}. 
Following are some simple reduction rules for {\sc RIAFD}:

\begin{definition}\label{red:afd1}
    \textbf{Reduction $1$:} Apply the following rules exhaustively, until the remaining graph has minimum vertex degree at least $2$:
    \begin{enumerate}
        \item Delete all vertices of degree at most one in the input graph.
        \item If $k<0$, then we have a no instance. If $k > 0$ \emph{and} $G$ is an $\ell$-forest, we have a yes instance. If $k=0$, we have a yes instance iff $G$ is an $\ell$-forest.
    \end{enumerate}
\end{definition}

\subsubsection{Dense Case}
Now we give a probabilistic reduction for {\sc RIAFD} that capitalizes on the fact that a large number of edges are incident to the {\sf riafd-set}. In particular, for a yes instance we focus on obtaining a probabilistic reduction that succeeds with probability strictly greater than $1/3$ so as to achieve a randomized algorithm running in time $\mathcal{O}^{\star}(3-\epsilon)^{k}$ with high probability.

\begin{definition} \label{red:afd2}
    \textbf{Reduction $2$ (P):} Assume that Reduction 1 does not apply and $G$ has a vertex of degree at least $3$. Sample a vertex $v \in V$ proportional to $\omega(v) := (deg(v)-2)$ if $v \notin R$, else $\omega(v) := 0$. That is, select each vertex with probability $\frac{\omega(v)}{\omega(V)}$. Delete $v$ and add its neighbours to $R$. Decrease $k$ by $1$.
\end{definition}

\begin{claim} \label{claim:afd_degf}
    Let $G$ be a graph, $F$ an {\sf afd-set} of $G$. Denote $\overline{F} := V \setminus F$. We have that,
    $$deg(\overline{F}) \le deg(F) + 2(|\overline{F}|-1+\ell) $$
\end{claim}
\begin{claimproof}
    This proof is based on simple observations. Notice that $deg(\overline{F}) = 2E(\overline{F})+ E(\overline{F},F)$. As $G[\overline{F}]$ is an $\ell$-forest, $E(\overline{F}) \le |\overline{F}|-1+\ell$. Also, $E(\overline{F},F)\le deg(F)$. Therefore, 
    $$deg(\overline{F})\le 2(|\overline{F}|-1+\ell)+deg(F)$$
\end{claimproof}

\begin{lemma} \label{lem:afd_red2}
    Given a graph $G$, if there exists a {\sf riafd-set} $F$ of size $k$ such that $deg(F) \ge \frac{4-2\epsilon}{1-\epsilon}(k+\ell)$,  then success of Reduction 2, which is essentially sampling a vertex $v\in F$ occurs with probability at least $\frac{1}{3-\epsilon}$.
\end{lemma}
\begin{proof}
    Let $F \in V$ is a {\sf riafd-set} of $G$ of size exactly $k$. For Reduction $2$ to succeed with probability at least $\frac{1}{3-\epsilon}$, we need $\frac{\omega(F)}{\omega(\overline{F})} \ge \frac{1}{2-\epsilon}$.
    
    The value of $\omega(F)$ can be rewritten as,
    $$\omega(F) = \sum_{v \in F} (deg(v) - 2) = deg(F)-2k.$$
    By Claim \ref{claim:afd_degf} (as {\sf riafd-set} is also an {\sf afd-set}),
    $$\omega(\overline{F}) \le \sum_{v \in \overline{F}}(deg(v)-2) = deg(\overline{F}) - 2|\overline{F}| \le deg(F) + 2(|\overline{F}|- 1 + \ell) - 2|\overline{F}| \le deg(F) + 2\ell.$$
    
    Therefore, 
    $$\frac{\omega(F)}{\omega(\overline{F})} \ge \frac{deg(F)-2k}{deg(F)+2\ell} = 1 - \frac{2(k+\ell)}{deg(F)+2\ell} \operatorname*{\ge}\limits^{(\ell\ge 0)} 1 - \frac{2(k+\ell)}{deg(F)}. $$
    
    Hence, we need
    $$ 1-\frac{2(k+\ell)}{deg(F)} \ge \frac{1}{2-\epsilon} \iff deg(F) \ge \frac{4-2\epsilon}{1-\epsilon}(k+\ell).$$
\end{proof}

\subsubsection{Sparse Case}
For the sparse case, we first construct a small separator. Due to the presence of two variables ($k$ and $\ell$), we have to modify the small separator lemma in \cite{li-soda20} with a bivariate analysis. Also, though we are discussing {\sc RIAFD}, we will show how to construct a small separator assuming that we are given an {\sf afd-set}, as a {\sf riafd-set} is also an {\sf afd-set}.

\paragraph{Small Separator}
The main idea, as presented in \cite{li-soda20}, is to convert an {\sf afd-set} with small average degree into a good tree decomposition. In particular, suppose a graph $G$ has an {\sf afd-set} $F$ of size $k$ with $deg(F) \le \overline{d}(k+\ell)$, where $\overline{d} = \mathcal{O}(1)$. We show how to construct a tree decomposition of width $(1 - \Omega(1))k + (2 - \Omega(1))\ell$. Note that $\overline{d}$ is not exactly the average degree of $F$. This definition helps us to bound the width of the tree decomposition well.

Before constructing this separator, we will first see a construction of a $\beta$-separator of an $\ell$-Forest. We could use Lemma \ref{lem:generalized_beta_separator}, but the size of the separator obtained would be $\ell\cdot o(k)$ which is huge (treewidth $\le \ell$). We now give a method to construct a $\beta$-separator of size $\ell + o(k)$.

\begin{lemma}\label{lem:afd_betasep}
    Given an $\ell$-forest $T(V,E)$ on $n$ vertices with vertex weights $\omega(v)$, for any $\beta > 0$, we can delete a set $S$ of $\beta + \ell$ vertices in polynomial time so that every connected component of $T - S$ has total weight at most $\frac{\omega(V)}{\beta}$.
\end{lemma}
\begin{proof}
    Construct some spanning tree for each connected component of $T$, call this resultant forest $T'$. Let $X$ be the set of remaining edges which are not in $T'$. For each edge in $X$, delete one vertex from $T'$. As $|X| \le \ell$, we will delete at most $\ell$ vertices. The resultant will still be a forest, call it $T''$.
    
    Now, root every component of the forest $T''$ at an arbitrary vertex. Iteratively select a vertex $v$ of maximal depth whose subtree has total weight more than $\frac{\omega(V)}{\beta}$, and then remove $v$ and its subtree. The subtrees rooted at the children of $v$ have total weight at most $\frac{\omega(V)}{\beta}$, since otherwise, $v$ would not satisfy the maximal depth condition. Moreover, by removing the subtree rooted at $v$, we remove at least $\frac{\omega(V)}{\beta}$ total weight, and this can only happen $\beta$ times. Thus, we delete at most $\beta + \ell$ vertices overall.
\end{proof}

With the help of Lemma \ref{lem:afd_betasep}, we will now proceed to the small separator lemma.

\begin{lemma}\label{lem:afd_smallsep}{\rm (Small Separator).} Given an instance $(G, k, \ell)$ and an {\sf afd-set} $F$ of $G$ of size $k$, define $\overline{d} := \frac{deg(F)}{k+\ell}$, and suppose that $\overline{d} = \mathcal{O}(1)$. There is a randomized algorithm running in expected polynomial time that computes a separation $(A, B, S)$ of $G$ such that:
    \begin{enumerate}
        \item $|A \cap F|, |B \cap F| \ge (2^{-\overline{d}}-o(1))(k+\ell) - \ell$
        \item $|S| \le (1 + o(1))(k+\ell) - |A \cap F| - |B \cap F|$
    \end{enumerate}
\end{lemma}

\begin{proof}
The proof will be similar to \cite{li-soda20} (Lemma $4$). First, we fix a parameter $\epsilon:= (k+\ell)^{-0.01}$ throughout the proof. Apply Lemma \ref{lem:afd_betasep} to the $\ell$-forest $G - F$ with $\beta = \epsilon (k+\ell)$ and vertex $v$ weighted by $|E[v, F]|$. Let $S_{\epsilon}$ be the output. Observe that:
$$|S_{\epsilon}| \le \ell+\epsilon(k+\ell) = \ell + o(k+\ell),$$
and every connected component $C$ of $G - F - S_{\epsilon}$ satisfies,
$$|E[C, F]| \le \frac{|E[\overline{F},F]|}{\epsilon(k+\ell)} \le \frac{deg(F)}{\epsilon(k+\ell)} \le \frac{\overline{d}(k+\ell)}{\epsilon(k+\ell)} = \frac{\overline{d}}{\epsilon}$$

Now form a bipartite graph $H$, as in \cite{li-soda20}, i.e., on the vertex bipartition $F \uplus \mathcal{R}$, where $F$ is the {\sf afd-set}, and there are two types of vertices in $\mathcal{R}$, the component vertices and the subdivision vertices. For every connected component $C$ in $G - F - S_{\epsilon}$, there is a component vertex $v_C$ in $\mathcal{R}$ that represents that component, and it is connected to all vertices in $F$ adjacent to at least one vertex in $C$. For every edge $e = (u, v)$ in $E[F,F]$, there is a vertex $v_e$ in $\mathcal{R}$ with $u$ and $v$ as its neighbours. Observe that: 
\begin{itemize}
    \item $|\mathcal{R}| \le |E[\overline{F}, F]| + 2|E[F,F]| = deg(F).$
    \item every vertex in $\mathcal{R}$ has degree at most $\frac{\overline{d}}{\epsilon}.$
    \item the degree of a vertex $v \in F$ in $H$ is at most $deg(v).$.
\end{itemize}
The algorithm that finds a separator $(A,B,S)$ works as follows. For each vertex in $\mathcal{R}$, color it red or blue uniformly and independently at random. Every component $C$ in $G - F - S_{\epsilon}$ whose vertex $v_C$ is colored red is added to $A$ in the separation $(A, B, S)$, and every component whose vertex $v_C$ is colored blue is added to $B$. Every vertex in $F$ whose neighbors are all colored red joins $A$, and every vertex in $F$ whose neighbors are all colored blue joins $B$. The remaining vertices in $F$, along with the vertices in $S_{\epsilon}$, comprise $S$. It is easy to see that $(A,B,S)$ is a separation.

\begin{claim}
    $(A,B,S)$ is a separation.
\end{claim}

We now show with good probability both conditions $(1)$ and $(2)$ hold. The algorithm can then repeat the process until both conditions hold.

\begin{claim}
    With probability at least $1-\frac{1}{k^{\mathcal{O}(1)}}$ condition $(1)$ holds for $(A,B,S)$.
\end{claim} 
\begin{claimproof}
Firstly, notice that $F$ has at most $\epsilon (k+\ell)$ vertices with degree at least $\frac{\overline{d}}{\epsilon}$. These can be ignored as they affect condition $(1)$ only by an additive $\epsilon(k+\ell) = o(k+\ell)$ factor. Let $F'$ be the vertices with degree at most $\frac{\overline{d}}{\epsilon}$. Now, consider the \emph{intersection graph} $I$ on vertices of $F'$ formed by connecting two vertices if they share a common neighbour (in $\mathcal{R}$). Since every vertex in $F'$ and all the component vertices have degree at most $\frac{\overline{d}}{\epsilon}$, the maximum degree of $I$ is at most $\left(\frac{\overline{d}}{\epsilon}\right)^2$. Color the vertices of $F'$ with $\left(\frac{\overline{d}}{\epsilon}\right)^2+1$ colors such that the vertices of the same color class form an independent set in $I$, using the standard greedy algorithm. Note that, within each color class, the outcome of each vertex whether it joins $A, B$ or $S$ is independent across vertices.

Let $F'_i$ be the set of vertices colored $i$. If $|F'_i| \le k^{0.9}$, then this color class can be ignored since the sum of all such $|F'_i|$ is at most $\left(\left(\frac{\overline{d}}{\epsilon}\right)^2+1\right)k^{0.9} = o(k)$ and this affects condition $(1)$ by an additive $o(k)$ factor. Henceforth, assume $|F'_i| \ge k^{0.9}$. Each vertex $v \in F'_i$ has at most $deg(v)$ neighbours in $H$. So, it can join $A$ with an independent probability of at least $2^{-deg(v)}$. Let $X_i = |F'_i \cap A|$, then by Hoeffding's inequality \footnote{We use the notation $exp(x)$ to denote the function $e^x$.},
$$\Pr[X_i \le E[X_i] - k^{0.8}]\le 2\cdot exp\left(-2\cdot\frac{\left(k^{0.8}\right)^2}{|F'_i|}\right) \le 2\cdot exp\left(-2\cdot\frac{k^{1.6}}{k}\right) \le \frac{1}{k^{\mathcal{O}(1)}}$$
for large enough $k$.

By a union bound over all the $\le k^{0.1}$ such color classes with $|F'_i| \ge k^{0.9}$, the probability that $|F'_i \cap A| \ge E[|F'_i \cap A|] - k^{0.8}$ for each $F'_i$ is $1 - \frac{1}{k^{\mathcal{O}(1)}}$. In this case,
\begin{eqnarray*}
|F \cap A| & \ge & \sum\limits_{i:|F'_i|\ge k^{0.9}} \left(E[|F'_i \cap A|] - k^{0.8}\right)\\
 & \ge & \sum\limits_{i:|F'_i|\ge k^{0.9}} \sum\limits_{v\in F'_i} \left(2^{-deg(v)}\right) - k^{0.1}\cdot k^{0.8}\\
 & = & \sum\limits_{v \in F'} 2^{-deg(v)} + \sum\limits_{j = 1}^{\ell} 2^{0} - \ell - o(k)\\
 & \ge & \left(|F'|+\ell\right)\cdot 2^{-\frac{deg(F')}{|F'|+\ell}} - \ell - o(k),
\end{eqnarray*}
where the last inequality follows from convexity of the function $2^{-x}$. Recall that $|F'| \ge k - o(k+\ell)$, and observe that $\frac{deg(F')}{|F'|+\ell} \le \frac{deg(F)}{k+\ell} = \overline{d}$ since the vertices in $F \setminus F'$ are vertices with degree greater than some threshold. Thus,
$$|F \cap A| \ge \left(k+\ell - o(k+\ell)\right)\cdot 2^{-\overline{d}} - l - o(k) \ge \left(2^{-\overline{d}} - o(1)\right)(k+\ell) - \ell,$$
proving condition $(1)$ for $A$. The argument for $|B \cap F|$ is symmetric.
\end{claimproof}

\begin{claim}
    With probability at least $1-\frac{1}{k^{\mathcal{O}(1)}}$ condition $(2)$ holds for $(A,B,S)$.
\end{claim}
\begin{claimproof}
    Note that at most $\ell + o(k+l)$ vertices in $S$ are from $S_{\epsilon}$, and the other vertices in $S$ are from the set $F \setminus ((F\cap A) \cup (F \cap B))$ which has size $k - |A \cap F|-|B \cap F|$. Thus, the overall size of $S \le \left(1+o(1)\right)(k+\ell)-|A\cap F|-|B\cap F|$
\end{claimproof}
\end{proof}

\begin{lemma}\label{lem:afd_smallsep_final}
     Let $G$ be a graph and $F$ be a {\sf afd-set} of $G$ of size $k$, and define $\overline{d} := \frac{deg(F)}{k+\ell}$. There is a randomized algorithm that, given $G$ and $F$, computes a tree decomposition of $G$ of width at most $(1-2^{-\overline{d}}+o(1))k + (2-2^{-\overline{d}}+o(1))\ell$, and runs in polynomial time in expectation.
\end{lemma}
\begin{proof}
    Compute a separation $(A, B, S)$ following Lemma \ref{lem:afd_smallsep}. Conditions $(1)$ and $(2)$ can be checked easily in polynomial time, so we can repeatedly compute a separation until they both hold.
    
    Notice that $G[A \cup S] - (F \cup S)$ is a forest, as $S_{\epsilon}$ includes the $\ell$ vertices corresponding to the $\ell$ extra edges of the $\ell$-Forest $G-F$. Thus, $(A \cap F) \cup S$ is a {\sf fvs} of $A \cup S$. The size of this {\sf fvs} is,
    $$|(A \cap F) \cup S| = |A \cap F|+ |S| \le (1+o(1))(k+\ell) - |B \cap F| \le (1 - 2^{-\overline{d}}+o(1))k + (2 - 2^{-\overline{d}}+o(1))\ell.$$
    
    Therefore, we can compute a tree decomposition of $G[A \cup S]$ of width $(1 - 2^{-\overline{d}}+o(1))k + (2 - 2^{-\overline{d}}+o(1))\ell$ as follows: start with a tree decomposition of width $1$ of the forest $G[A\cup S] - (F \cup S)$, and then add all vertices in $(A \cap F) \cup S$ to each bag. Call this tree decomposition of $G[A \cup S]$ as $\mathbb{T}_1$. Similarly, compute a tree decomposition of $G[B\cup S]$ in the same way, call it $\mathbb{T}_2$. 
    
    Since there is no edge connecting $A$ to $B$, we can construct the tree decomposition $\mathbb{T}$ of $G$ by simply adding an edge between an arbitrary node from $\mathbb{T}_1$ and $\mathbb{T}_2$. It is evident from the construction procedure that $\mathbb{T}$ is a valid tree decomposition of $G$ and it takes polynomial time to compute it.
\end{proof}

\begin{note}\label{note:afd_smallsep}
    Using the tree decomposition obtained in Lemma \ref{lem:afd_smallsep_final}, we can run the \emph{Cut \& Count} algorithm from Section \ref{sec:afd_cutcount}. But this will utilize exponential space. To get polynomial space, we use the idea from Claim \ref{claim:afdpoly}.
    In the proof of Lemma \ref{lem:afd_smallsep_final}, observe that the set $(A \cap F) \cup S$ is present in every bag of $\mathbb{T}_1$. Similarly, $(B\cap F) \cup S$ is present in every bag of $\mathbb{T}_2$. This observation is crucial for the proof of Lemma \ref{lem:afd_bafd1}.
\end{note}

As we are in the sparse case, there exists a {\sf riafd-set} $F$ of size $k$ with bounded degree, i.e., $deg(F) \le \overline{d}k$. We call this bounded version of the problem {\sc BRIAFD}. As we saw, the small separator helps in constructing a tree decomposition of small width, but requires that we are given an {\sf afd-set} of size $k$ and bounded degree. To attain this, we use an \emph{Iterative Compression} based procedure which at every iteration constructs a {\sf riafd-set} of size at most $k$ with bounded degree and uses it to construct the small separator. Using this small separator we construct a tree decomposition of small width and run a \emph{Cut \& Count} based procedure to solve bounded {\sc RIAFD} problem for the current induced subgraph, i.e, get a {\sf riafd-set} of size at most $k$ with bounded degree.

Now, we give the claimed \hyperref[alg:bafd1]{\texttt{BRIAFD1}} algorithm, which is a \emph{Cut \& Count} based algorithm which solves {\sc BRIAFD} given a small separator.

\begin{algorithm}[!ht]
\label{alg:bafd1}
\DontPrintSemicolon
\caption{\texttt{BRIAFD1}$(G,R,k,\ell,F,A,B,S,\overline{d})$}
\SetKwInOut{Input}{Input}\SetKwInOut{Output}{Output}
\SetKwData{afdcc}{RIAFDCutandCount}
\SetKwData{afdfcc}{RIAFD-FCCount}
\SetKwData{c}{count}
\SetKwData{ca}{countA}
\SetKwData{cb}{countB}
\SetKw{Break}{break}
\SetKwData{flag}{flag}
\SetKwFunction{tdp}{treewidthDP}
\SetKwData{infs}{Infeasible}
\SetKw{Return}{return}
\SetKwData{Heads}{Heads}
% \SetKwFunction{IC}{\hyperref[ic]{IC}}
\SetKwFunction{BAFD}{\hyperref[bafd1]{BRIAFD1}}
\Input{Graph $G(V,E)$, a set $R$, an {\sf afd-set} $F$ of size at most $k+1$, the parameters $k, \overline{d} \le n$ and $\ell \le m$ and a separation $(A,B,S)$ from Lemma \ref{lem:afd_smallsep}.}
\Output{Either output a {\sf riafd-set} $F^\star$ of size at most $k$ satisfying $deg(F^\star) \le \overline{d}(|F^\star|+\ell)$, or conclude that one does not exist (\infs).}
\Begin{
    \For{$A \ge n - k$, $0 \le B \le A+\ell-1$, $W = i|V|^2 + d$ for some $d \le \overline{d}(n-A+\ell)$}{
        $t \leftarrow A - B + \ell + 1$\;
        \For{$n^{\mathcal{O}(1)}$ iterations}{
            \c $\leftarrow 0$\;
            Randomly initialize $\omega$ as stated in Lemma \ref{lem:afdccmain}\;
            Generate tree decompositions $\mathbb{T}_1$ and $\mathbb{T}_2$ as defined in proof of Lemma \ref{lem:afd_smallsep_final}\;
            \For{all possible $f: S \rightarrow \left\{\text{\textbf{F}},\text{\textbf{L}},\text{\textbf{R}}\right\}$}{
                \For{$W'$, $A'$, $B'$ such that $0\le W' \le W$, $0 \le A' \le A$, $0 \le B' \le B$}{
                \ca $\leftarrow 0$\;
                    \For{all possible $f_A: (A \cap F) \rightarrow \left\{\text{\textbf{F}},\text{\textbf{L}},\text{\textbf{R}}\right\}$}{
                        \ca $\leftarrow$ \ca $+$ $\mathtt{RIAFD-FCCount}$($\mathbb{T}_1$,$R$,$A'$,$B'$,$W'$,$f \uplus f_A$,$t$)\;
                    }
                    \cb $\leftarrow 0$\;
                    \For{all possible $f_B: (B \cap F) \rightarrow \left\{\text{\textbf{F}},\text{\textbf{L}},\text{\textbf{R}}\right\}$}{
                        \cb $\leftarrow$ \cb $+$ $\mathtt{RIAFD-FCCount}$($\mathbb{T}_2$,$R$,$A - A' +$\; $|f^{-1}\left(\{\text{\textbf{L}},\text{\textbf{R}}\}\right)|$,$B - B' + \left|E\left[f^{-1}\left(\{\text{\textbf{L}},\text{\textbf{R}}\}\right)\right]\right|$,$W - W' + \omega(f^{-1}(\textbf{F}))$,$f \uplus f_B$,$t$)\;
                    }
                    \c $\leftarrow$ \c $+$ \ca $.$ \cb\;
                }
            }
            \If{\c $\not\equiv 0 \ (\textrm{mod}\ 2^t)$}{
                $F^\star \leftarrow$ a {\sf riafd-set} of $G$ of size at most $k$ satisfying  $deg(F^\star) \le \overline{d}(|F^\star|+\ell)$ constructed using self-reduction\;
                \Return $F^\star$\;
            }
    }
}
\Return \infs\;

}
\end{algorithm}

\begin{lemma}\label{lem:afd_bafd1}
    Given a graph $G$, a set $R$, an {\sf afd-set} $F$ of $G$ of size at most $k+1$, parameter $\overline{d}$, and a separation $(A, B, S)$ as given by Lemma \ref{lem:afd_smallsep}, the Algorithm \hyperref[alg:bafd1]{\texttt{BRIAFD1}} outputs an {\sf riafd-set} $F^{\star}$ of size at most $k$ satisfying $deg(F^{\star}) \le \overline{d}(|F^{\star}|+\ell)$, or Infeasible if none exists. The algorithm uses $\mathcal{O}^{\star}(3^{(1-2^{-\overline{d}}+o(1))k}\cdot 3^{(2-2^{-\overline{d}}+o(1))\ell})$ time and polynomial space and succeeds with high probability.
\end{lemma}
\begin{proof}
    For the time bound, firstly notice that lines $12$ and $15$ take polynomial time due to the observation given in Note \ref{note:afd_smallsep} and Claim \ref{claim:afdpoly}. All other steps listed in the algorithm \texttt{BRIAFD1} are polynomial time except lines $8$, $11$ and $14$, which jointly give rise to $3^{|S|+|A \cap F|} + 3^{|S|+|B \cap F|}$ iterations. By the conditions $(1)$ and $(2)$ of the separator $(A,B,S)$ in Lemma \ref{lem:afd_smallsep}, we get the desired time bound. The space bound is evident from the description of the algorithm \texttt{BRIAFD1} and Claim \ref{claim:afdpoly}. Also, by Line $4$ and Lemma \ref{lem:randomalgo}, the algorithm succeeds with high probability.
    
    For the correctness, first we claim that at Line $16$, count $= |\mathcal{C}_W^{n-k,B}|$ for some $A$, $B$ and $W = in^2 + d$ (from Lemma \ref{lem:afdccmain}). To see the claim, observe that we are iterating over all possible mappings of $S$. For each mapping and every possible split of the parameters $W$ and $B$, the algorithm computes the number countA (resp. countB) denoting the ``extensions'' of the mapping in $G[A \cup S]$ (resp. $G[B \cup S]$) that respect the split, and then multiplies countA and countB. To see why these counts are multiplied, notice that there are no edges between $A$ and $B$. So, extending into $G[A \cup S]$ is independent to extending into $G[B \cup S]$. This along with the correctness of \texttt{RIAFD-FCCount} proves the claim, thereby proving the correctness.
\end{proof}

And now we give the Iterative Compression routine \hyperref[alg:afdic1]{\texttt{RIAFD\_IC1}}, as explained above, which solves {\sc BRIAFD}. 

\begin{algorithm}[!ht]
\label{alg:afdic1}
\DontPrintSemicolon
\caption{$\mathtt{RIAFD\_IC1}$$(G,R,k,\ell,\overline{d})$}
\SetKwInOut{Input}{Input}\SetKwInOut{Output}{Output}
\SetKwComment{Comment}{$\triangleright$ Invariant: }{}
\SetKwData{infs}{Infeasible}
\SetKwFunction{bafd}{BRIAFD1}
\SetKw{Return}{return}
\Input{Graph $G = (V,E)$, a set $R$ and parameters $k,\overline{d} \le n$ and $\ell \le m$ where $\overline{d} = \mathcal{O}(1)$.}
\Output{A {\sf riafd-set} $F^\star$ of $G$ of size at most $k$ satisfying $deg(F^\star) \le \overline{d}(|F^\star|+\ell)$ or \infs.}
\Begin{
Order the vertices $V$ in ascending order of degrees and call them $(v_1, v_2, \ldots, v_n)$ \;
$F^\star \leftarrow \varnothing$\;
\For{$i = 1, 2 \ldots, n$}{
    \Comment*[h]{$deg(F^\star) \le \overline{d}(|F^\star| + \ell)$}\;
    Compute a separation $(A,B,S')$ of $G\left[\left\{v_1, \ldots, v_{i-1}\right\}\right]$ by Lemma \ref{lem:afd_smallsep} using $F^\star, \overline{d}$\;
    $S \leftarrow S' \cup \left\{v_i\right\}$ so $(A,B,S)$ is a separation of $G\left[\left\{v_1, \ldots, v_{i}\right\}\right]$\;
    $F^\star \leftarrow$ \bafd($G$,$R, k, \ell, F^\star \cup \{v_i\}, A, B, S, \overline{d}$)\;
    \If{$F^\star$ is \infs}{
        \Return \infs\;
    }
    }

\Return $F^\star$\;
}
\end{algorithm}

\begin{lemma}\label{lem:afd_IC1}
    Algorithm \hyperref[alg:afdic1]{$\mathtt{RIAFD\_IC1}$} solves {\sc BRIAFD} in $\mathcal{O}^{\star}(3^{(1 - 2^{-\overline{d}}+o(1))k} \cdot 3^{(2 - 2^{-\overline{d}}+o(1))\ell})$ time and polynomial space.
\end{lemma}

\begin{proof}
    Suppose there exists a {\sf riafd-set} $F$ of size $k$ satisfying $deg(F) \le \overline{d}(k+\ell)$. Let $\left(v_1, \ldots, v_i\right)$ be the ordering from Line $2$, and define $V_i := \left\{v_1, \ldots, v_i\right\}$. Observe that $F \cap V_i$ is a {\sf riafd-set} of $G[V_i]$ of size at most $k$. Let $F_i = F \cap V_i$ and $|F_i| = k_i \le k$. Due to the ordering from Line $2$, $F_i$ are the vertices of least degrees in $F$. Thus, $\frac{deg(F_i)}{k_i+\ell} \le \frac{deg(F)}{k+l} \le \overline{d}$. Hence, {\sc BRIAFD} problem on Line $7$ is feasible.
    
    Line $7$ correctly computes a bounded degree {\sf riafd-set} of size at most $k$ with high probability, by Lemma \ref{lem:afd_bafd1}. Therefore, with high probability, a {\sf riafd-set} of size $k$ is returned.

     We now bound the running time. On Line $5$, the current set $F^\star$ is a {\sf riafd-set} of $G[V_{i-1}]$ satisfying $deg(F^\star) \le \overline{d}(|F^\star|+\ell)$, so Lemma \ref{lem:afd_smallsep_final} guarantees tree decompositions $\mathbb{T}_1$ and $\mathbb{T}_2$ of width at most $(1 - 2^{-\overline{d}} + o(1))k + (2 - 2^{-\overline{d}} + o(1))\ell$, and adding $v_i$ to each bag on Line $6$ increases the width by at most $1$. By Lemma \ref{lem:afd_bafd1}, Line $7$ runs in time $\mathcal{O}^{\star}\left(3^{(1-2^{-\overline{d}}+o(1))k}\cdot 3^{(2-2^{-\overline{d}}+o(1))\ell}\right)$, as desired. The space bound is evident from the description of \texttt{RIAFD\_IC1} and Lemma \ref{lem:afd_bafd1}.
\end{proof}

\paragraph{Three-Way Separator}
Similar to small separator, a bivariate analysis has to be done in the case of the Three-Way separator too. The outline of the analysis is similar to Lemma \ref{lem:afd_smallsep}.

\begin{lemma}\label{lem:afd_3waysep} {\rm (Three-Way Separator).} Given an instance $(G, k)$ and an {\sf afd-set} $F$ of size $k$, define $\overline{d} := \frac{deg(F)}{k+\ell}$, and suppose that $\overline{d} = \mathcal{O}(1)$. There is a randomized algorithm running in expected polynomial time that computes a three-way separation $(S_1, S_2, S_3, S_{1,2}, S_{1,3}, S_{2,3}, S_{1,2,3})$ of $G$ such that there exists values $f_1, f_2$ satisfying:
    \begin{enumerate}
        \item $f_1 k \ge (3^{-\overline{d}}-o(1))(k+\ell)-\ell$
        \item $f_1 k - o(k+\ell) \le |S_i \cap F| \le f_1 k + o(k+\ell)$ for all $i \in [3]$ 
        \item $(f_2 + 2f_1)k \ge \big((\frac{2}{3})^{\overline{d}} - o(1)\big)(k+\ell) - \ell$
        \item $f_2 k - o(k+\ell) \le |S_{i,j}| \le f_2 k + o(k+\ell)$ for all $1 \le i < j \le 3$
    \end{enumerate}
\end{lemma}
\begin{proof}
    This proof is similar to \cite{li-soda20}(Lemma $14$) and uses the idea from Lemma \ref{lem:afd_smallsep}. Firstly, we start out the same: fix $\epsilon := (k+\ell)^{-0.01}$, apply Lemma \ref{lem:afd_betasep} on $G - F$ (to construct $S_\epsilon$), and construct the bipartite graph $H$ on the bipartition $F \uplus \mathcal{R}$ in the same way as in Lemma \ref{lem:afd_smallsep}. Recall that,
    \begin{itemize}
        \item $|\mathcal{R}| \le |E[\overline{F}, F]| + 2|E[F,F]| = deg(F)$.
        \item every vertex in $\mathcal{R}$ has degree at most $\frac{\overline{d}}{\epsilon}$.
        \item the degree of a vertex $v \in F$ in $H$ is at most $deg(v)$.
    \end{itemize}
    Now, instead of randomly two-coloring the vertex set $\mathcal{R}$, the algorithm three-colors it. That is, for each vertex in $\mathcal{R}$, color it with a color in $\{1, 2, 3\}$ chosen uniformly and independently at random. For each subset $I \subseteq 2^{[3]}\setminus \{\varnothing\}$, create a vertex set $S_I$ consisting of all vertices $v \in F$ whose neighborhood in $H$ sees the color set $I$ precisely. More formally, let $c(v)$ and $N(v)$ be the color of $v \in \mathcal{R}$ and the neighbors of $v$ in $H$, and define $S_I = \{v \in F: \sum\limits_{u\in N(v)} c(u) = I\}$. Furthermore, if $I$ is a singleton set $\{i\}$, then add (to $S_I$) all vertices in the connected components $C$ whose component vertex in $\mathcal{R}$ is colored $i$. Henceforth, we abuse notation, referring to sets $S_{\{1\}}, S_{\{1,2\}},$ etc. as $S_1, S_{1,2},$ etc.
    
    \begin{claim}\label{claim:afd_3way_1}
        $(S_1, S_2, S_3, S_{1,2}, S_{1,3}, S_{2,3}, S_{1,2,3})$ is a three-way separator.
    \end{claim} 
    
    \begin{claim}\label{claim:afd_3way_2}
        For $f_1 := \frac{\sum_d p_d |F'_d|}{|F'|}$, condition $(2)$ holds with probability at least $1 - \frac{1}{k^{\mathcal{O}(1)}}$.
    \end{claim}
    
    \begin{claim}\label{claim:afd_3way_3}
        For $f_2 := \frac{\sum_d p_d |F'_d|}{|F'|}$, condition $(4)$ holds with probability at least $1 - \frac{1}{k^{\mathcal{O}(1)}}$.
    \end{claim}
    
   In Claim  \ref{claim:afd_3way_3}, $p_d$ is the probability of a vertex to join $S_{i,j}$ for any $i,j\in[3]$ such that $i\neq j$. The proofs of Claims \ref{claim:afd_3way_1}, \ref{claim:afd_3way_2} and \ref{claim:afd_3way_3} are very similar to the proofs in \cite{li-soda20} and the proof of Lemma \ref{lem:afd_smallsep}. Hence, they are omitted.
    
    \begin{claim}\label{claim:afd_3way_4}
        For $f_1 := \frac{\sum_d p_d |F'_d|}{|F'|}$, condition $(1)$ holds with probability at least $1 - \frac{1}{k^{\mathcal{O}(1)}}$
    \end{claim}
    \begin{claimproof}
        Here, $p_d$ is the probability of a vertex with degree $d$ to join $S_i$ for any $i \in [3]$. It's easy to see that $p_d = 3^{-d}$. Observe that $\frac{deg(F')}{|F'|+\ell} \le \frac{deg(F)}{k+\ell} = \overline{d}$, since the vertices in $F\setminus F'$ are precisely vertices with degree exceeding some threshold, and $|F'| \ge k - o(k+\ell)$. Also, due to the convexity of the function $3^{-x}$, we get       
        \begin{eqnarray*}
            f_1k \ge f_1|F'| & = & \sum\limits_{d} |F'_d|\cdot 3^{-d}\\
             & = & \sum\limits_{v \in F'} 3^{-deg(v)} + \sum\limits_{j = 1}^{\ell}3^{0} - \ell\\
             & \ge & (|F'| + \ell)3^{-\frac{deg(F')}{|F'|+\ell}} - \ell\\
             & \ge & (3^{-\overline{d}} - o(1))(k+\ell) - \ell,
        \end{eqnarray*}
        proving condition $(1)$.
    \end{claimproof}
    \begin{claim}\label{claim:afd_3way_5}
        For $f_1 := \frac{\sum_d p_d |F'_d|}{|F'|}$ and $f_2 := \frac{\sum_d p_d |F'_d|}{|F'|}$, condition $(3)$ holds.
    \end{claim}
    
    \begin{claimproof}
        Let $q_d$ be the probability that a vertex $v$ of degree $d$ joins one of $S_i$, $S_2$ or $S_{1,2}$. Since this is also the probability that no neighbour of $v$ is colored $3$, we have $q_d = \big(\frac{2}{3}\big)^{d}$. Let $p_{1,d}$ and $p_{2,d}$ be the probabilities $p_d$ in the Claims \ref{claim:afd_3way_2} and  \ref{claim:afd_3way_3} respectively, so that $q_d = 2p_{1,d} + p_{2,d}$. Therefore,
        
        \begin{eqnarray*}
            2f_1 k + f_2 k \ge 2f_1|F'| + f_2|F'| & = & 2\cdot \sum\limits_{d} p_{1,d}\cdot|F'_d| + \sum\limits_{d} p_{2,d}\cdot|F'_d|\\
             & = & \sum\limits_{d}|F'|\cdot q_d\\
             & = & \sum\limits_{v \in F'} \left(\frac{2}{3}\right)^{deg(v)} + \sum\limits_{j = 1}^{\ell}3^{0} - \ell\\
             & \ge & (|F'| + \ell)\left(\frac{2}{3}\right)^{\frac{deg(F')}{|F'|+\ell}} - \ell,
        \end{eqnarray*}
        where the last inequality follows from convexity of $\left(\frac{2}{3}\right)^x$. Again, we have $\frac{deg(F')}{|F'|+\ell} \le \frac{deg(F)}{k+\ell} = \overline{d}$, and $|F'| \ge k - o(k+\ell)$. So,
        $$(f_2 + 2f_1)k \ge \left(\left(\frac{2}{3}\right)^{\overline{d}} - o(1)\right)(k+\ell) - \ell,$$
        proving condition $(3)$.
    \end{claimproof}
\end{proof}

We now describe the structure of the three-way separator in more detail which will help in designing the algorithm utilizing it. Let's say we are given a graph $G(V,E)$, an {\sf afd-set} $F$ of size at most $k+1$ and a three-way separation $(S_1, S_2, S_3, S_{1,2}, S_{2,3}, S_{2,3},S_{1,2,3})$ as in Lemma \ref{lem:afd_3waysep}. Let $f_1$ and $f_2$ be from the conditions of Lemma \ref{lem:afd_3waysep}. Define $f_3:=1-3f_1-3f_2$, so that $f_3k +\ell - o(k+\ell) \le |S_{1,2,3}|\le f_3k + \ell+ o(k+\ell)$.
    
Notice that $G[S_1 \cup S_{1,2}\cup S_{1,3} \cup S_{1,2,3}] - (F \cup S_{1,2,3})$ is a forest, as $S_{\epsilon}$ (from Lemma~\ref{lem:afd_3waysep}) includes the $\ell$ vertices corresponding to the $\ell$ extra edges of the $\ell$-Forest $G-F$. Thus, $(S_1 \cap F)\cup S_{1,2}\cup S_{1,3} \cup S_{1,2,3}$ is an {\sf fvs} of $S_1 \cup S_{1,2}\cup S_{1,3} \cup S_{1,2,3}$. The size of this {\sf fvs} is,
$$|(S_1 \cap F)\cup S_{1,2}\cup S_{1,3} \cup S_{1,2,3}| = |S_1 \cap F|+ |S_{1,2}|+|S_{1,3}|+ |S_{1,2,3}| \le (f_3+2f_2+f_1)k+\ell+o(k+\ell)$$
    
Therefore, we can compute a tree decomposition of $G[S_1 \cup S_{1,2}\cup S_{1,3} \cup S_{1,2,3}]$ of width $(f_3+2f_2+f_1)k+\ell+o(k+\ell)$ as follows: start with a tree decomposition of width $1$ of the forest $G[S_1 \cup S_{1,2}\cup S_{1,3} \cup S_{1,2,3}] - (F \cup S_{1,2,3})$, and then add all vertices in $(S_1 \cap F) \cup S_{1,2}\cup S_{1,3} \cup S_{1,2,3}$ to each bag. Call this tree decomposition $\mathbb{T}_1$. Similarly, we can compute a tree decomposition of $G[S_2 \cup S_{1,2}\cup S_{2,3} \cup S_{1,2,3}]$ and $G[S_3 \cup S_{1,3}\cup S_{2,3} \cup S_{1,2,3}]$ in the same way, call them $\mathbb{T}_2$ and $\mathbb{T}_3$ respectively. It is evident from the construction procedure it takes polynomial time to compute these tree decompositions.
    
\begin{note}\label{note:afd_3waysep}
    Observe that there is no edge connecting any pair among $S_1$, $S_2$ and $S_3$, and $S_{i,j}$ has neighbours only in $S_i$ and $S_j$. Also, the set $(S_1 \cap F) \cup S_{1,2}\cup S_{1,3} \cup S_{1,2,3}$ is present in every bag of $\mathbb{T}_1$. Similarly, $(S_2 \cap F) \cup S_{1,2}\cup S_{2,3} \cup S_{1,2,3}$ and $(S_3 \cap F) \cup S_{1,3}\cup S_{2,3} \cup S_{1,2,3}$ are present in every bag of $\mathbb{T}_2$ and $\mathbb{T}_3$ respectively. This observation and the three decompositions obtained will be crucial for the proof of Lemma \ref{lem:afd_bafd2}.
\end{note}

Similar to the two-way separator case, we now give the routines \hyperref[alg:bafd2]{\texttt{BRIAFD2}} and \hyperref[alg:afdic2]{\texttt{RIAFD\_IC2}} which will utilize the three-way separator.

\begin{algorithm}[!ht]
\DontPrintSemicolon
\caption{\texttt{BRIAFD2}$(G,R,k,\ell,F,S_1, S_2, S_3, S_{1,2}, S_{1,3}, S_{2,3}, S_{1,2,3},\overline{d})$}
\label{alg:bafd2}
\algorithmfootnote{Values of some variables are not assigned in the pseudocode above to maintain clarity. In the algorithm $w$, $a$, $b$ are variables to account for overcounting in $S_{1,2,3}$. If we define $s_1 = f^{-1}(\{\text{\textbf{L}}, \text{\textbf{R}}\})$ then $w = 2 \cdot \omega(S_{1,2,3} \setminus s_1)$, $a = 2 \cdot |s_1|$ and $b = 2 \cdot |E[s_1,s_1]|$. For the overcounting that takes place within $S_{1,2}$, $S_{2,3}$ and $S_{3,1}$ we define the variables $w_i$, $a_i$ and $b_i$ for $i \in [3]$. We take $w_1 = a_1 = b_1 = 0$. If we define $s_2 = f_2^{-1}(\{\text{\textbf{L}}, \text{\textbf{R}}\})$, then $w_2 = \omega(S_{1,2} \setminus s_2)$, $a_2 = |s_2|$, $b_2 = |E[s_2,s_2]|$. If we define $s_3 = f_1^{-1}(\{\text{\textbf{L}}, \text{\textbf{R}}\}) \uplus f_2^{-1}(\{\text{\textbf{L}}, \text{\textbf{R}}\})$, then $w_3 = \omega((S_{2,3} \uplus S_{3,1}) \setminus s_3)$, $a_3 = |s_3|$, $b_3 = |E[s_3,s_3]|$.}
\SetKwInOut{Input}{Input}\SetKwInOut{Output}{Output}
\SetKwData{afdcc}{RIAFDCutandCount}
\SetKwData{afdfcc}{RIAFD-FCCount}
\SetKwData{c}{count}
\SetKwData{ct}{count3}
\SetKwData{cz}{count0}
\SetKw{Break}{break}
\SetKwData{flag}{flag}
\SetKwFunction{tdp}{treewidthDP}
\SetKwData{infs}{Infeasible}
\SetKw{Return}{return}
\SetKwData{Heads}{Heads}
% \SetKwFunction{IC}{\hyperref[ic]{IC}}
\SetKwFunction{BAFD}{\hyperref[alg:bafd2]{BRIAFD2}}
\Input{Graph $G = (V,E)$, a set $R$, an {\sf afd-set} of size at most $k+1$, the parameters $k, \overline{d} \le n$ and $\ell \le m$ and a separation $(S_1,S_2,S_3,S_{1,2},S_{1,3},S_{2,3},S_{1,2,3})$ from Lemma \ref{lem:afd_3waysep}.}
\Output{Either output a {\sf riafd-set} $F^\star$ of size at most $k$ satisfying $deg(F^\star) \le \overline{d}(|F^\star|+\ell)$, or conclude that one does not exist (\infs).}
\Begin{
    \For{$A \ge n - k$, $0 \le B \le A+\ell-1$, $W = i|V|^2 + d$ for some $d \le \overline{d}(n-A+\ell)$}{
        $t \leftarrow A - B + \ell + 1$\;
        \For{$n^{\mathcal{O}(1)}$ iterations}{
            \c $\leftarrow 0$\;
            Randomly initialize $\omega$ as stated in Lemma \ref{lem:afdccmain}\;
            Generate tree decompositions $\mathbb{T}_1$, $\mathbb{T}_2$ and $\mathbb{T}_3$ as stated in Note \ref{note:afd_3waysep}\;
            \For{all possible $f: S_{1,2,3} \rightarrow \left\{\text{\textbf{F}},\text{\textbf{L}},\text{\textbf{R}}\right\}$}{
                \For{nonnegative $W_i$, $A_i$, $B_i$, $i \in [3]$ such that $\sum_iW_i = W + w$, $\sum_iA_i = A + a$, $\sum_iB_i = B + b$}{
                    $H \leftarrow$ an empty graph with vertices indexed by $ \binom{S_{1,2}}{.,.,.}\cup\binom{S_{2,3}}{.,.,.} \cup \binom{S_{3,1}}{.,.,.}$\;
                    \For{$(i,j,k)$ in $\left\{(1,2,3),(2,3,1),(3,1,2)\right\}$}{
                        \For{all possible $f_1: S_{i,j} \rightarrow \left\{\text{\textbf{F}},\text{\textbf{L}},\text{\textbf{R}}\right\}$, $f_2: S_{i,k} \rightarrow \left\{\text{\textbf{F}},\text{\textbf{L}},\text{\textbf{R}}\right\}$}{
                            \ct $\leftarrow 0$\;
                            \For{all possible $f_3: S_{i} \cap F \rightarrow \left\{\text{\textbf{F}},\text{\textbf{L}},\text{\textbf{R}}\right\}$}{
                                \ct $\leftarrow$ \ct $+$ $\mathtt{RIAFD-FCCount}$($\mathbb{T}_i$, $R$, $A_i + a_i$, $B_i + b_i$, $W_i + w_i$, $f \uplus f_1 \uplus f_2\uplus f_3$)\;
                            }
                            Add edge $e$ between vertices $\left(f_1^{-1}(\text{\textbf{F}}),f_1^{-1}(\text{\textbf{L}}), f_1^{-1}(\text{\textbf{R}})\right)$ and $\left(f_2^{-1}(\text{\textbf{F}}),f_2^{-1}(\text{\textbf{L}}), f_2^{-1}(\text{\textbf{R}})\right)$ of $H$\;
                            Assign weight \ct $(\textrm{mod}\  2^t)$ to edge $e$\;
                        }
                    }
                    \cz $\leftarrow$ sum over the product of the three edges of all triangles of $H$\;
                    \c $\leftarrow$ \c $+$ \cz\;
                }
            }
            \If{\c $\not\equiv 0 \ (\textrm{mod}\ 2^t)$}{
                $F^\star \leftarrow$ a {\sf riafd-set} of $G$ of size $\le k$ satisfying  $deg(F^\star) \le \overline{d}(|F^\star|+\ell) $ constructed using self-reduction\;
                \Return $F^\star$\;
            }
        }
    }
\Return \infs\;
}
\end{algorithm}

\begin{lemma}\label{lem:afd_bafd2}
    Given a graph $G$, an {\sf afd-set} $F$ of $G$ of size at most $k+1$, parameter $\overline{d}$, and a three-way separation $(S_1, S_2, S_3, S_{1,2}, S_{1,3}, S_{2,3},S_{1,2,3})$ as given by Lemma \ref{lem:afd_3waysep}, the Algorithm \hyperref[alg:bafd2]{$\mathtt{BRIAFD2}$} outputs a {\sf riafd-set} of size at most $k$ satisfying $deg(F) \le \overline{d}(|F|+\ell)$, or Infeasible if none exists. The algorithm uses $\mathcal{O}^{\star}(3^{(1-min\{\left(\frac{2}{3}\right)^{\overline{d}},(3-\omega)\left(\frac{2}{3}\right)^{\overline{d}}+(2\omega-3)3^{-\overline{d}}\}+o(1))k}\cdot 3^{(1+\omega-((3-\omega)\left(\frac{2}{3}\right)^{\overline{d}}+(2\omega-3)3^{-\overline{d}})+o(1))\ell})$ time.
\end{lemma}
\begin{proof}
    For the time bound, firstly notice that lines $15$ takes polynomial time due to the observation given in Note \ref{note:afd_3waysep} and Claim \ref{claim:afdpoly}. Let $f_1,f_2$ and $f_3$ be from Lemma \ref{lem:afd_3waysep} and Note \ref{note:afd_3waysep}. For each of the $\mathcal{O}^{\star}(3^{f_3k + \ell + o(k+\ell)})$ iterations on Line $8$, building the graph $H$ (Lines $10-17$) takes time $\mathcal{O}^{\star}(3^{(2f_2+f_1)k+o(k+\ell)})$, and running matrix multiplication on Line $18$ on a graph with $\mathcal{O}^{\star}(3^{f_2k+o(k+\ell)})$ vertices to compute the sum over product of weights on the three edges of all triangles takes time $\mathcal{O}^{\star}(3^{\omega f_2 k + o(k+\ell)})$. Therefore, the total running time is
    
    \hspace{0.2cm}$\mathcal{O}^{\star}(3^{f_3k+\ell+o(k+\ell)}(3^{(2f_2+f_1)k+o(k+\ell)} + 3^{\omega f_2 k + o(k+\ell)}))$
    \begin{eqnarray*}
        &=&\mathcal{O}^{\star}(3^{(f_3+2f_2+f_1)k+\ell+o(k+\ell))} + 3^{(f_3+\omega f_2)k+l+o(k+\ell)})\\
        &=&\mathcal{O}^{\star}(3^{(1-f_2-2f_1)k+\ell+o(k+\ell))} + 3^{(1-(3-\omega) f_2-3f_1)k+l+o(k+\ell)})\\
        &=&\mathcal{O}^{\star}(3^{(1-(f_2+2f_1))k+\ell+o(k+\ell))} + 3^{(1-(3-\omega)(f_2+2f_1)-(2\omega-3)f_1)k+l+o(k+\ell)})\\
        &=&\mathcal{O}^{\star}(3^{(1-(f_2+2f_1))k+\ell+o(k+\ell))} + 3^{(1-(3-\omega)(f_2+2f_1)-(2\omega-3)f_1)k+l+o(k+\ell)})\\
        &\le &\mathcal{O}^{\star}((3^{1-(\frac{2}{3})^{\overline{d}}+o(1))k} + 3^{(1-((3-\omega)(\frac{2}{3})^{\overline{d}}+(2\omega-3)3^{-\overline{d}}+o(1))k})\cdot\\
        & & \:\:\:\:\:\:\:\:\:3^{(1+\omega-((3-\omega)(\frac{2}{3})^{\overline{d}}+(2\omega-3)3^{-\overline{d}})+o(1))\ell}),
    \end{eqnarray*}
    where the last inequality uses the conditions $(1)$ and $(3)$ of Lemma \ref{lem:afd_3waysep}, and the fact that $2\omega -3 \ge 0$. This gives the desired time bound.
    
    The proof of correctness is similar to proof of Lemma $15$ in \cite{li-soda20}. We claim that at Line $19$, count $= |\mathcal{C}_W^{n-k,B}|$ for some $A$, $B$ and $W = in^2 + d$ (from Lemma \ref{lem:afdccmain}). First observe that there is no edge between $S_1$ and $S_{2,3}$. So, number of extensions of $S_1$ only depend on $S_{1,2}$ and $S_{1,3}$. For each mapping of $S_{1,2} \cup S_{1,3}$, imagine adding an edge between the respective mappings in the graph $H$, with weight as the number of extensions in $S_1$. Proceed analogously in $S_2$ and $S_3$. Thus, $H$ will be a tripartite graph. Now, merging the solutions, i.e. finding the total number of extensions (for a fixed mapping of $S_{1,2,3}$), amounts to computing the sum over product of weights of three edges forming triangles in $H$, which can be solved using a standard matrix multiplication routine. This along with correctness of \texttt{RIAFD-FCCount} completes the proof of the claim, thereby completing the proof of correctness.
\end{proof}

\begin{algorithm}[!ht]
\label{alg:afdic2}
\DontPrintSemicolon
\caption{$\mathtt{RIAFD\_IC2}$$(G,R,k,\ell,\overline{d})$}
\SetKwInOut{Input}{Input}\SetKwInOut{Output}{Output}
\SetKwComment{Comment}{$\triangleright$ Invariant: }{}
\SetKwData{infs}{Infeasible}
\SetKwFunction{bafd}{BRIAFD2}
\SetKw{Return}{return}
\Input{Graph $G = (V,E)$, a set $R$ and parameters $k,\overline{d} \le n$ and $\ell \le m$ where $\overline{d} = \mathcal{O}(1)$.}
\Output{A {\sf riafd-set} $F^\star$ of size at most $k$ satisfying $deg(F^\star) \le \overline{d}(|F^\star|+\ell)$ or \infs.}
\Begin{
Order the vertices $V$ in ascending order of degrees and call them $(v_1, v_2, \ldots, v_n)$ \;
$F^\star \leftarrow \varnothing$\;
\For{$i = 1, 2, \ldots, n$}{
\Comment*[h]{$deg(F^\star) \le \overline{d}(|F^\star| + \ell)$}\;
    Compute a separation $(S_1,S_2, S_3, S_{1,2}, S_{1,3}, S_{2,3},S'_{1,2,3})$ of $G\left[\left\{v_1, \ldots, v_{i-1}\right\}\right]$ by Lemma \ref{lem:afd_3waysep} for given $F^\star$, $\overline{d}$\;
    $S_{1,2,3} \leftarrow S'_{1,2,3} \cup  \left\{v_i\right\}$, so $(S_1,S_2, S_3, S_{1,2}, S_{1,3}, S_{2,3},S_{1,2,3})$ is a three-way separation of $G\left[\left\{v_1, \ldots, v_{i}\right\}\right]$\;
    $F^\star \leftarrow$ \bafd($G, R, k, \ell, F^\star \cup \{v_i\}, S_1,S_2, S_3, S_{1,2}, S_{1,3}, S_{2,3},S_{1,2,3}, \overline{d}$)\;
    \If{$F^\star$ is \infs}{
        \Return \infs\;
    }
    }
    \Return $F^\star$\;
    
}
\end{algorithm}

\begin{lemma}\label{lem:afd_IC2}
    Algorithm \hyperref[alg:afdic2]{$\mathtt{RIAFD\_IC2}$} solves {\sc BRIAFD} in\\    $\mathcal{O}^{\star}(3^{(1-min\{(\frac{2}{3})^{\overline{d}},(3-\omega)(\frac{2}{3})^{\overline{d}}+(2\omega-3)3^{-\overline{d}}\}+o(1))k}\cdot 3^{(1+\omega-((3-\omega)(\frac{2}{3})^{\overline{d}}+(2\omega-3)3^{-\overline{d}})+o(1))\ell})$ time.
\end{lemma}
The proof is similar to the proof of Lemma \ref{lem:afd_IC1}, hence it is omitted.

\subsubsection{Algorithms for {\sc RIAFD}}
Having described the \emph{Dense} and the \emph{Sparse} Cases, we now combine them to give the final randomized algorithms.

\paragraph{$2.85^k8.54^\ell$ Algorithm in Polynomial Space}
Now, we give the Algorithm \hyperref[alg:afd1]{\texttt{RIAFD1}$(G,k,\ell)$}, which is the complete randomized algorithm combining the Dense and the Sparse Cases (small separator).

\begin{algorithm}[!ht]
\label{alg:afd1}
\DontPrintSemicolon
\caption{\texttt{RIAFD1}$(G,R,k,\ell)$}
\SetKwInOut{Input}{Input}\SetKwInOut{Output}{Output}
\SetKwFunction{tdp}{treewidthDP}
\SetKwData{infs}{Infeasible}
\SetKw{Return}{return}
\SetKwData{Heads}{Heads}
\SetKwFunction{IC}{\hyperref[alg:afdic1]{RIAFD\_IC1}}
\SetKwFunction{AFD}{\hyperref[alg:afd1]{RIAFD1}}
\Input{Graph $G = (V,E)$, a set $R$, two parameters $k \le n$ and $\ell \le m$.}
\Output{Either output a {\sf riafd-set} $F$ of size at most $k$, or (possibly incorrectly) conclude that one does not exist (\infs).}
\Begin{
\For{$0 \le k' \le k$}{
Exhaustively apply Reduction 1 to $(G,R,k',\ell)$ to get vertex set $F'$ and the instance $(G',R,k'',\ell)$\;
$\overline{d} \leftarrow (4 - 2\epsilon)/(1 - \epsilon)$\;
Flip a coin with \Heads probability $3^{-(1 - 2^{\overline{d}})k''} \cdot 3^{-(2 - 2^{\overline{d}})\ell}$\;
\If{coin flipped \Heads}{
$F \leftarrow$ \IC{$G'$,$R$,$k''$,$\ell$, $\overline{d}$}\;}
\Else{
Apply Reduction 2 to $(G',R', k'', \ell)$ to get vertex $v\in V$ and instance $(G'',R'',k'' - 1,\ell)$\;
$F \leftarrow$ \AFD{$G''$,$R''$, $k''-1$, $\ell$} $\cup \left\{v\right\}$\;
}
\If{$F \cup F'$ is not \infs}{\Return $F \cup F'$}
}
\Return \infs}
\end{algorithm}

\begin{lemma}\label{lem:afd_AFD1}
Fix the parameter $\epsilon \in (0,1)$ and $\overline{d}:=\frac{4-2\epsilon}{1-\epsilon}$, let $c_k:=max\big{\{}3-\epsilon,3^{1-2^{-\overline{d}}}\big{\}}$ and $c_\ell:=3^{2-2^{-\overline{d}}}$. Then \hyperref[alg:afd1]{$\mathtt{RIAFD1}$$(G,k,\ell)$} succeeds with probability at least $\frac{c_k^{-k}c_\ell^{-\ell}}{k}$ and has $\mathcal{O}^{\star}(3^{o(k+\ell)})$ expected running time.
\end{lemma}

\begin{proof}
    We will focus on running time for each iteration of the outer loop. The computation till line $6$ takes $n^{\mathcal{O}(1)}$ time. For each $k'' \in (0, k']$, Line $7$ is executed with probability $3^{-(1 - 2^{\overline{d}})k'} \cdot 3^{-(2 - 2^{\overline{d}})\ell}$ and takes time $\mathcal{O}^{\star}(3^{(1 - 2^{\overline{d}}+o(1))k'} \cdot 3^{(2 - 2^{\overline{d}}+o(1))\ell})$. So, in expectation, the total computation cost of Line $7$ is $\mathcal{O}^{\star}(3^{o(k+\ell)})$ per value of $k''$, and also $\mathcal{O}^{\star}(3^{o(k+\ell)})$ overall. Note here that for all values of $\epsilon \in (0,1)$,  $c_k\ge2$ and $c_l\ge1$.   
    
    Now, we prove that \hyperref[alg:afd1]{\texttt{RIAFD1}$(G,k,\ell)$} succeeds with probability $\frac{c_k^{-k}\cdot c_\ell^{-\ell}}{k}$. For simplicity of calculations, we replace $k'$ with $k$. Moreover, as each iteration is an independent trial, $k$ is an upper bound for any $k'$ that succeeds. We use Induction on $k$. The statement is trivial when $k=0$, since no probabilistic reduction is used and hence it succeeds with probability $1$. For the inductive step, consider an instance \hyperref[alg:afd1]{\texttt{RIAFD1}$(G,k+1,\ell)$}. Let $(G', k'',\ell)$ be the reduced instance after Line $3$. Suppose that every {\sf riafd-set} $F$ of $G$ of size $k''$ satisfies the condition $deg(F) \le \overline{d}(k''+l)$; here, we only need the existence of one such $F$. In this case, if Line $7$ is executed, then it will correctly output a {\sf riafd-set} $F$ of size at most $k''$, with high probability by Lemma \ref{lem:afd_IC1}. This happens with probability at least    
    $$3^{-(1-2^{\overline{d}})k''} \cdot 3^{-(2-2^{\overline{d}})\ell} \cdot \left(1 - \frac{1}{n^{\mathcal{O}(1)}}\right) \ge c_k^{-k''}\cdot c_\ell^{-\ell} \cdot\frac{1}{k} \ge \frac{c_k^{-k}\cdot c_\ell^{-\ell}}{k},$$
    as desired.
    
    Otherwise, suppose that the above condition doesn't hold for every {\sf riafd-set} $F$ of $G'$ of size $k''$. This means that there exists a {\sf riafd-set} $F$ of size $k''$ such that $deg(F) \ge \overline{d}(k''+l)$. In this case, by Lemma \ref{lem:afd_red2}, Reduction $2$ succeeds with probability at least $\frac{1}{3-\epsilon}$. This is assuming, of course, that Line $7$ is not executed, which happens with probability $1 - c_k^{-k''}\cdot c_\ell^{-\ell} \ge 1 - c_k^{-k''} \ge 1 - 2^{-k''} \ge 1 - \frac{1}{k''}$, since $c_l \ge 1$ and $c_k \ge 2$. By Induction, the recursive call on Line $10$ succeeds with probability at least $\frac{c_k^{-(k''-1)}\cdot c_\ell^{-\ell}}{(k''-1)}$. So, the overall probability of success is at least,
    $$\left(1-\frac{1}{k''}\right)\cdot\frac{1}{3-\epsilon}\cdot \frac{c_k^{-(k''-1)}\cdot c_\ell^{-\ell}}{(k''-1)} \ge \left(\frac{k''-1}{k'}\right)\cdot\frac{1}{c_k}\cdot \frac{c_k^{-(k''-1)}\cdot c_\ell^{-\ell}}{(k''-1)} = \frac{c_k^{-k''}\cdot c_\ell^{-\ell}}{k''} \ge \frac{c_k^{-k}\cdot c_\ell^{-\ell}}{k},$$
    as desired. Note that on line 10 adding the neighbours of $v$ to $R'$ in the recursive call ensures that $F$ is independent on addition of $v$ to it.
\end{proof}

Unless $R$ is explicitly nonempty, we set $R=\varnothing$ to solve {\sc RIAFD}. To optimize for $c_k$, we set $\epsilon \approx 0.155433$, giving $c_k \le 2.8446$ and $c_\ell \le 8.5337$. Theorem \ref{thm:afd} (2) now follows by combining Lemma \ref{lem:afd_AFD1} and Lemma \ref{lem:randomalgo}.

\paragraph{$2.7^k36.61^\ell$ Algorithm using Matrix Multiplication}
Using Lemma \ref{lem:afd_bafd2} and Lemma \ref{lem:afd_IC2} and the Dense Case, we now prove the main result, Theorem \ref{thm:afd} $(3)$, restated below.\\

\noindent
$\blacktriangleright$ \textbf{Theorem \ref{thm:afd} (3).} There is a randomized algorithm that solves {\sc RIAFD} in time $\mathcal{O}^{\star}(2.7^k36.61^\ell)$, with high probability.

\begin{proof}
    We run \hyperref[alg:afd1]{\texttt{RIAFD1}}, replacing every occurrence of \texttt{RIAFD\_IC1} with \texttt{RIAFD\_IC2}. We define $\overline{d} := \frac{4-2\epsilon}{1 - \epsilon}$ for some $\epsilon > 0$ (to be determined later); note that $\overline{d} \ge 4$ for any $\epsilon > 0$. Since $\omega < 2.3728639$ \cite{matmul}, by Lemma \ref{lem:afd_IC2}, \texttt{RIAFD\_IC2} runs in time $\mathcal{O}^{\star}(3^{(1-((3-\omega)(\frac{2}{3})^{\overline{d}}+(2\omega-3)3^{-\overline{d}})+o(1))k}\cdot 3^{(1+\omega-((3-\omega)(\frac{2}{3})^{\overline{d}}+(2\omega-3)3^{-\overline{d}})+o(1))\ell})$. Hence, \hyperref[alg:afd1]{\texttt{RIAFD1}} runs in time $\mathcal{O}^{\star}(c_k^k\cdot c_\ell^\ell)$, by Lemma \ref{lem:randomalgo} to get high success probability, for $c_k := max\big{\{} 3-\epsilon, 3^{1-((3-\omega)(\frac{2}{3})^{\overline{d}}+(2\omega-3)3^{-\overline{d}})}\big{\}}$ and \\
    $c_\ell := 3^{1+\omega-((3-\omega)(\frac{2}{3})^{\overline{d}}+(2\omega-3)3^{-\overline{d}})}$. For $\omega = 2.3728639$, we optimize for $c_k$ and set $\epsilon \approx 0.3000237$, giving $c_k \le 2.699977$ and $c_\ell \le 36.602$, as desired. If $\omega=2$, we can always substitute it's value and optimize on $c_k$ to get the values for $c_k$ and $c_\ell$, if required.
\end{proof}

\subsection{Improving Dependence on $\ell$}\label{sec:afd_improve_l}
In this subsection, we will try to reduce the dependence on $\ell$ in the \emph{Cut \& Count} algorithm. To achieve this, we will construct a tree decomposition with reduced dependence on $\ell$.

\begin{lemma}
\label{lem:afd_tw_l}
    Given a graph $G(V,E)$ with $tw(G)>2$ and a {\sf riafd-set} $F$ of size $k$, there exists a tree decomposition of width $k + \frac{3}{5.769}\ell + \mathcal{O}(\log(\ell))$ for $G$ and it can be constructed in polynomial time. 
\end{lemma}
\begin{proof}
    Given $G(V,E)$ with $n$ vertices and $m$ edges, we define the graph $G'(V',E') := G[V/F]$. $G'$ is an $\ell$-Forest from the definition of {\sf riafd-set}. We apply the following reduction rules exhaustively on $G'$:
    \begin{itemize}
        \item $R_0$:  If there is a $v\in V'$ with $deg(v)=0$, then remove $v$.
        \item $R_1$: If there is a $v\in V'$ with $deg(v)=1$, then remove $v$.
        \item $R_2$: If there is a $v\in V'$ with $deg(v)=2$, then contract $v$, i.e.  remove $v$ and insert a new edge between its two neighbors, if no such edge exists.
    \end{itemize}
    
    For the safeness of the above reduction rules refer to \cite{kneis}. Let the reduced graph be called $G''(V'',E'')$. It is trivial to see that after applying these rules the $G''$ we get is also an $\ell$-Forest. Therefore, after removing at most $\ell$ edges from $G''$, we are left with at most $|V''| - 1$ edges (since the remaining graph is a forest). Therefore, we get that $|E''| \le |V''|+\ell-1$. Since the degree of each vertex in $G''$ is at least $3$, $|E''| \ge 3|V''|/2$. Therefore, $1.5|V''| \le |V''| + \ell - 1$ from which we obtain the bounds $|V''| \le 2\ell$ and $|E''| \le 3\ell$.  We need to use the following results  
    from~\cite{kneis}.
    
   \begin{theorem} {\rm \cite[Theorem 4.7]{kneis}.} Given a graph ${G}(V,E)$, we can obtain a tree decomposition of ${G}$ of width at most $|E|/5.769 + \mathcal{O}(log(|V|))$ in polynomial time. 
    \end{theorem}
    
    This implies that, $G''$ has a tree decomposition of width at most $\frac{3}{5.769}\ell + \mathcal{O}(\log(\ell))$ which can be computed in polynomial time. 
    
      \begin{lemma} {\rm \cite[Lemma 4.2]{kneis}.} Given a connected graph $G$, with $\tw(G)>2$ and let $G'$ be a graph obtained from $G$ by applying $R_0$, $R_1$ and $R_2$ then $\tw(G) = \tw(G')$    \end{lemma}
    
%    \textbf{Lemma 4.2 in \cite{kneis}:} Given a connected graph $\mathcal{G}$, with $tw(\mathcal{G})>2$ and let $\mathcal{G}'$ be a graph obtained from $G$ by applying $R_0$, $R_1$ and $R_2$ then $tw(\mathcal{G}) = tw(\mathcal{G}')$.\\
    Also, from proof of Lemma 4.2 of \cite{kneis}, it's easy to see that this also works on graphs which might not be connected. Given these facts, we see that we can obtain a tree decomposition of $G'$ with width at most $\frac{3}{5.769}\ell + \mathcal{O}(\log(\ell))$ in polynomial time from the tree decomposition of $G''$. Now to get the tree decomposition of the given graph instance $G$, add $F$ (of size $k$ which we removed) to all the bags of the tree decomposition of $G'$. This finally gives the required tree decomposition of $G$ of width at most $k + \frac{3}{5.769}\ell + \mathcal{O}(\log(\ell))$.
\end{proof}

We combine the treewidth bound that can be obtained from Lemma \ref{lem:afd_tw_l} with \emph{Iterative Compression}, together with the $3^{\tw}$ algorithm to obtain an $\mathcal{O}^\star(3^k1.78^\ell)$ algorithm for solving {\sc RIAFD}. 

We now describe the working of the routine \hyperref[alg:afd_ic_l]{\texttt{RIAFD\_IC3}}. The iterative compression routine proceeds as follows. We start with an empty graph, and add the vertices of $G$ one by one, while always maintaining a {\sf riafd-set} of size at most $k$ in the current graph. Maintaining a {\sf riafd-set} for the current graph helps us utilize Lemma \ref{lem:afd_tw_l} to obtain a small tree decomposition (of size $k + \frac{3}{5.769}\ell + \mathcal{O}(\log(\ell))$). Then we can add the next vertex in the ordering to all the bags in the tree decomposition to get a new {\sf riafd-set} of size $k$ in $\mathcal{O}^\star(3^{\tw})$. If we are unable to find such a {\sf riafd-set} in a particular iteration, we can terminate the algorithm early.

\begin{algorithm}[!ht]
\label{alg:afd_ic_l}
\DontPrintSemicolon
\caption{$\mathtt{RIAFD\_IC3}$$(G,R,k,\ell)$}
\SetKwInOut{Input}{Input}\SetKwInOut{Output}{Output}
\SetKwData{tdp}{RIAFDCutandCount}
\SetKwData{infs}{Infeasible}
\SetKw{Return}{return}
\Input{Graph $G = (V,E)$, a set $R$ and parameters $k \le n$ and $\ell = \mathcal{O}(n^2)$.}
\Output{A {\sf riafd-set} $F$ of size at most $k$ or \infs.}
\Begin{
Order the vertices $V$ arbitrarily as $(v_1, v_2, \ldots, v_n)$ \;
$F \leftarrow \varnothing$\;
\For{$i = 1, 2, \ldots, n$}{
    $\mathbb{T} \leftarrow$ Compute the tree decomposition of $G\left[\left\{v_1, \ldots, v_{i-1}\right\}\right]$ by Lemma \ref{lem:afd_tw_l}\;
    Add $v_i$ to all bags of $\mathbb{T}$\;
    $F \leftarrow$ a {\sf riafd-set} of $G\left[\left\{v_1, \ldots ,v_i\right\}\right]$ with parameters $k$ and $\ell$, computed using $\mathtt{RIAFDCutandCount}$ on $\mathbb{T}$\;
    \If{$F$ is \infs}{
        \Return \infs\;
    }}
\Return $F$\;
}
\end{algorithm}

Now we restate Theorem \ref{thm:afd} (4) and prove it.

\noindent 
 $\blacktriangleright$ \textbf{Theorem \ref{thm:afd} (4).} {\em \hyperref[alg:afd_ic_l]{$\mathtt{RIAFD\_IC3}$} solves {\sc RIAFD} problem in time $\mathcal{O}^\star(3^k1.78^\ell)$ and exponential space with high probability. }

\begin{proof} 
Suppose that there exists a {\sf riafd-set} $F^\star$ of size at most $k$. Let $(v_1, v_2, \ldots, v_n)$ be the ordering from Line $2$, and define $V_i:= \left\{v_1, \ldots,v_i\right\}$. We note that $F^\star \cap V_i$ is a {\sf riafd-set} of $G\left[V_i\right]$ so {\sc RIAFD} problem on Line $7$ will be feasible in each iteration (and will be computed correctly with high probability in every iteration). Therefore, with high probability, a {\sf riafd-set} is returned successfully (by union bound). 

We now bound the running time. On Line $5$, the current set $F$ is a {\sf riafd-set} of $G\left[V_i\right]$, so Lemma \ref{lem:afd_tw_l} guarantees a tree decomposition of width at most $k + \frac{3}{5.769}\ell + \mathcal{O}(\log(\ell))$ and adding $v_i$ to each bag on Line 6 increases the width by at most one. By the \emph{Cut \& Count} algorithm from Section \ref{sec:afd_cutcount}, Line 6 runs in time $\mathcal{O}^\star(3^{\left(k + \frac{3}{5.769}\ell + \mathcal{O}(\log(\ell))\right)}) = \mathcal{O}^\star(3^{\left(k + \frac{3}{5.769}\ell\right)})$ (since $\ell = \mathcal{O}(|V|^2)$ for non-trivial instance). This gives the desired time of $\mathcal{O}^\star(3^k1.78^\ell)$ on simplification. The space bound follows directly from the description of \hyperref[alg:afd_ic_l]{\texttt{RIAFD\_IC3}}, Lemma \ref{lem:afd_tw_l} and the space bound of the \emph{Cut \& Count} algorithm.
\end{proof}

% \section{Independent Feedback Vertex Set}
% \input{ifvs.tex}

\section{Pseudoforest Deletion}\label{sec:pseudoforestdel}
%!TEX root = main.tex
%{\sc Pseudoforest Deletion} is an ``extension'' of the {\sc Feedback Vertex Set} problem. We say that a graph is a \emph{Pseudotree} if it is either a tree or obtained by adding an edge to a tree. A graph is a \emph{Pseudoforest} if it's connected components are \emph{pseudotrees}.
%
%\defparprob{{\sc Pseudoforest Deletion}}{Given an undirected (multi)graph $G=(V,E)$ and an integer $k$.}{$k$}{Does there exist a subset $Y \subseteq V$ of size at most $k$ such that $G - Y$ is a pseudoforest.}
%
% \begin{definition}
% \textsc{Pseudoforest Deletion}
% \\\textbf{Input:} Given an undirected (multi)graph $G=(V,E)$ and an integer $k$.
% \\\textbf{Question:} Does there exist a subset $Y \subseteq V$ of size at most $k$ such that $G[V \setminus Y]$ is a pseudoforest.
% \end{definition}
In this section we present faster randomized algorithms for \textsc{Pseudoforest Deletion}.  In Section \ref{sec:pseudo_3tw} we present an $\mathcal{O}^{\star}(3^{\tw})$ \emph{Cut \& Count} algorithm building on techniques from \cite{cutandcount} for \fvs. Using this we give an $\mathcal{O}^\star(3^k)$ time and polynomial space algorithm in Section \ref{sec:pseudo_3k}. In Section \ref{sec:pseudo_improvedk}, we use the method in \cite{li-soda20} to get an $\mathcal{O}^\star(2.85^k)$ time and polynomial space algorithm. Henceforth, the abbreviation {\sf pds} denotes a pseudoforest deletion set, i.e., a solution to an instance of \textsc{Pseudoforest Deletion}. 
\subsection{$\mathcal{O}^{\star}(3^{\tw})$ Algorithm}\label{sec:pseudo_3tw}
\begin{lemma}
\label{pseu1}
A graph $G=(V,E)$ with $n$ vertices and $m$ edges is a pseudoforest if and only if it has $n-m$ connected components which are trees.
\end{lemma}

\begin{proof}
We only consider cases where $n\ge m$. Note that any graph $G(V,E)$ with $n$ vertices and $m$ edges has at least $n-m$ connected components which are trees. This is because of a simple additive argument and the fact that for a connected component other than a tree with $n'$ vertices and $m'$ edges, the term $n'-m'\le 0$. 

\textbf{Forward Direction: } If $G$ is a pseudoforest, then its connected components can be either a tree or a tree plus an edge. For the ``tree plus edge component'', $n'-m'=0$. Hence we have $n-m$ trees.

\textbf{Reverse Direction:} Conversely, assume for contradiction that $G$ has greater than $n-m$ trees. By pigeonhole principle at least one of the other components has $n'-m'\le -1$, a contradiction.
\end{proof}

We present a \emph{Cut \& Count} technique similar to the one for \fvs in \cite{cutandcount}. As the universe we take $U=V\times\{\boldsymbol{P},\boldsymbol{M}_1\}+E\times\{\boldsymbol{M}_2\}$. The main difference between our algorithm from the one for \fvs is we account for additional $\boldsymbol{M}_2$ markers for the edges. For each edge, we \emph{a priori} decide one of its endpoints to represent the edge. Also, given a set of marked edges $M_2$, $\psi(M_2)$ denotes the set of representative vertices of the edges in $M_2$. When an edge is marked, it is assumed to be deleted and it's representative vertex is marked. This assumption will be crucial in our algorithm.

We assign weights uniformly at random to the elements of our universe with the weight function $\omega:U\rightarrow\{1,\hdots,N\}$, where $N=2|U|=4|V|+2|E|$.

\smallskip

\noindent 
\textbf{The Cut Part.} For integers $A$, $B$, $C$, $D$, $W$ we define:
\begin{enumerate}
    \item ${\mathcal{R}^{A,B,C,D}_W}$ to be the family of solution candidates:  ${\mathcal{R}^{A,B,C,D}_W}$ is the family of triples $(X,M_1,M_2)$ where $X \subseteq{V}$, $|X|=A$, $|E(G[X])| = B+D$ of which $D$ edges are marked, i.e $M_2\subseteq E(G[X])$ and $|M_2|=D$, $M_1\subseteq X$, $|M_1|=C$ and $\omega((V \setminus X)\times\boldsymbol{\{P\}})+\omega(M_1\times\{\boldsymbol{M}_1\})+
    \omega(M_2\times\{\boldsymbol{M}_2\})=W$.\\
    \item ${S^{A,B,C,D}_W}$ to be the set of solutions: the family of triples $(X,M_1,M_2)$, where $(X,M_1,M_2)\in {\mathcal{R}^{A,B,C,D}_W}$ and every connected component of $G[X] - M_2$ is a tree containing at least one $\boldsymbol{M}_1$ or $\boldsymbol{M}_2$ marker.\\
    \item ${\mathcal{C}^{A,B,C,D}_W}$ to be the family of pairs $((X,M_1,M_2),(X_L,X_R))$ where $(X,M_1,M_2)\in{\mathcal{R}^{A,B,C,D}_W}$, $M_1\subseteq X_L$, $\psi(M_2)\subseteq X_L$ and $(X_L,X_R)$ is a consistent cut of $G[X]$.
\end{enumerate}

According to \cite{cutandcount}, a consistent cut $(X_L, X_R)$ is one where there is no edge between the cuts. But, as we stated that an edge marked with a marker $\boldsymbol{M_2}$ is deleted, these edges are allowed to cross the cuts. But the representative vertex must belong to $X_L$ only.

\begin{lemma}
\label{count_pseu1}
The graph $G$ admits a pseudoforest deletion set of size $k$ iff there exist integers $B$, $D$, $W$ such that ${S^{n-k,B,n-k-B-D,D}_W}$ is nonempty.
\end{lemma}
\begin{proof}
\textbf{Forward direction:} Let $G$ have a {\sf pds} $P$ of size $k$. Then $G' = G[V\setminus P]=(V',E')$ is a pseudoforest with $n-k$ vertices. Let $G'$ have $D$ connected components which are ``a tree plus an edge'' and by Lemma \ref{pseu1} $G'$ has $n-k-B-D$ connected components which are trees, where $B = |E'| - D$. Then we can place one $\boldsymbol{M}_1$ marker each for all the tree components. Let $M_1$ be the set of these marked vertices. In each of the $D$ ``tree plus an edge components'', only one cycle exists. Choose any edge belonging to that cycle as an $\boldsymbol{M}_2$ marker. Thus, by definition, this edge is deleted making the component a tree. Also, as defined above, the representative vertex of the deleted edge is marked. Let $M_2$ be the set of all the marked edges. Also, let $W := \omega((V \setminus X)\times\{\boldsymbol{P}\})+\omega(M_1\times\{\boldsymbol{M}_1\})+\omega(M_2\times\{\boldsymbol{M}_2\})$. We now see that $(X,M_1,M_2)\in S_W^{n-k,B,n-k-B-D,D}$.

\textbf{Reverse direction:} We have that ${S^{n-k,B,n-k-B-D,D}_W}$ is non-empty for some integers $B$, $D$ and $W$. Let us consider some $(X,M_1,M_2)\in {S^{n-k,B,n-k-B-D,D}_W}$. Then, the graph $G[X]$ has $n-k$ vertices, $B+D$ edges and every connected component of $G[X] - M_2$ is a tree with exactly one marker, one of $\boldsymbol{M}_1$ or $\boldsymbol{M}_2$, by definition. Notice that if a tree component in $G[X]$ is marked by an $\boldsymbol{M}_2$ marker, then the number of unmarked tree components remains the same, as on marking an edge, the edge is deleted (by definition) marking it's representative vertex. Thus, on deletion we get two trees among which one is marked while the other is still unmarked. These unmarked tree components necessarily have to be taken care of by $\boldsymbol{M}_1$ markers. Therefore, the number of tree components has to be equal to the number of $\boldsymbol{M}_1$ markers, i.e. the number of tree components is exactly $n-k-B-D$. Therefore, by Lemma \ref{pseu1} $G[X]$ is a pseudoforest. 
\end{proof}
\begin{lemma}
\label{count_pseu2}
$|\mathcal{C}^{A,B,C,D}_W|\equiv |S^{A,B,C,D}_W| \: (\textrm{mod} \: 2)$.
\end{lemma}

\begin{proof}
Consider a triple $(X,M_1,M_2)$ in $\mathcal{R}^{A,B,C,D}_W$. If $G[X] - M_2$ has $c$ connected components without any marker($\boldsymbol{M}_1$ or $\boldsymbol{M}_2$), then it contributes $2^c$ to $|\mathcal{C}^{A,B,C,D}_W|$. Hence, if $c \ge 1$, the triple $(X,M_1,M_2)$ contributes $2^c \equiv 0\: (\textrm{mod}\: 2)$ to $|\mathcal{C}^{A,B,C,D}_W| \: (\textrm{mod}\:2)$. A triple $(X,M_1,M_2) \in S^{A,B,C,D}_W$ iff $G[X] - M_2$ has no unmarked connected components. Thus, it contributes $1\: (\textrm{mod} \: 2)$ to both $S^{A,B,C,D}_W$  and $\mathcal{C}^{A,B,C,D}_W$. Hence, $|\mathcal{C}^{A,B,C,D}_W|\equiv |S^{A,B,C,D}_W| \: (\textrm{mod}\: 2)$.  
\end{proof}

 \noindent
\textbf{The Count Part.} For $A,B,C,D,W,(X,M_1,M_2)\in {\mathcal{R}^{A,B,C,D}_W}$, there are $2^{cc(M_1,M_2,G[X])}$ consistent cuts $(X_L,X_R)$ such that $((X,M_1,M_2),(X_L,X_R))\in {\mathcal{C}^{A,B,C,D}_W}$ where $cc(M_1,M_2,G[X])$ denotes the number of connected components of $G[X]$ which do not contain any marker from either of $M_1$ or $M_2$. Hence for $C \le A-B$ by Lemma \ref{pseu1} we have that $|{S^{A,B,C,D}_W}|\equiv |{\mathcal{C}^{A,B,C,D}_W}| \ (\textrm{mod}\  2)$. 

Now we describe a dynamic programming procedure \texttt{CountC($\omega,A,B,C,D,W,\mathbb{T}$)}, that given a nice tree decomposition $\mathbb{T}$, weight function $\omega$ and integers $A,B,C,D,W$, computes $|\mathcal{C}^{A,B,C,D}_W|$ \textrm{mod} $2$.
For every bag $x\in \mathbb{T}$, $a\leq |V|$, $b\leq |V|$, $c\leq |V|$, $d\leq |V|$, $w\leq 3N|V|$ and $s\in \{\boldsymbol{F},\boldsymbol{L},\boldsymbol{R}\}^{B_x}$ (called the colouring), define

\begin{eqnarray*}
    \mathcal{R}_x(a,b,c,d,w) & = & \Big\{(X,M_1,M_2)\;\big|\; X\subseteq V_x\wedge \,|X|=a\wedge\; |E_x\cap E(G[X])|=b+d\wedge \\ & & M_1\subseteq X \wedge\; M_2\subseteq E_x\cap E(G[X]) \wedge\;
     |M_1|=c\wedge\; |M_2|=d\wedge \\ & & \omega((V\setminus X)\times \boldsymbol{\{P\} })+\omega(M_1\times  \{\boldsymbol{M}_1\})+\omega(M_2\times \{\boldsymbol{M}_2\})=w\Big\} \\
    \mathcal{C}_x(a,b,c,d,w) & = & \Big\{((X,M_1,M_2),(X_L,X_R))\;\big|\;(X,M_1,M_2)\in \mathcal{R}_x(a,b,c,d,w)\wedge \\ & & M_1\subseteq X_L\wedge \psi(M_2)\subseteq X_L\wedge (X,(X_L,X_R)) \text{ is a consistently cut}\\  & & \text{subgraph of }  G_x\Big\} \\
    A_x(a,b,c,d,w,s) & = & \Big|\Big\{((X,M),(X_L,X_R))\in C_x(a,b,c,d,w)\;\big| (s(v)=\boldsymbol{L}\implies v\in X_L)\wedge\\& & (s(v)=\boldsymbol{R}\implies v\in X_R)\wedge (s(v)=\boldsymbol{F}\implies v\notin X)\Big\}\Big| \\
\end{eqnarray*}

Note that we may assume $b\leq|V|$ and $d\leq|V|$ because the number of edges in a pseudoforest cannot exceed the number of vertices. The accumulators $a,b,c,d,w$ keep track of the number of vertices, edges of $X$, $M_1$ markers, $M_2$ markers and the target weight respectively. Hence $A_x(a,b,c,d,w,s)$ is the number of pairs in $\mathcal{C}_x(a,b,c,d,w)$ having a fixed interface with vertices in $B_x$. Note that we choose a vertex to be an $M_1$ marker in its respective forget bag. For the $M_2$ marker for an edge we make the choice in the introduce edge bag, where we decide to not include it in $G[X]$ if it is chosen as a $M_2$ marker. Also note that the endpoints in this case for this edge can be on opposite sides of the cut.\\
The algorithm computes $A_x(a,b,c,d,w,s)$ for each bag $x\in \mathbb{T}$ and for all reasonable values of $a,b,c,d,w$ and $s$. We now give the recurrence for $A_x(a,b,c,d,w,s)$ used by the dynamic programming algorithm. In order to simplify notation let $v$ be the vertex introduced and contained in an introduce bag, $(u,v)$ the edge introduced in an introduce edge bag with $u$ being the representative of the edge (i.e. $\psi(\{(u,v)\}) = \{u\}$), and let $y, z$ stand for the left and right child of $x$ respctively in $\mathbb{T}$ if present.
\begin{itemize}
    \item \textbf{Leaf bag:}
    \begin{align*}
        A_x(0,0,0,0,0,\varnothing)&=1
    \end{align*}
    \item \textbf{Introduce vertex bag:}
    \begin{align*}
        A_x(a,b,c,d,w,s\cup\{(v,\boldsymbol{F})\})&=A_y(a,b,c,d,w-\omega((v,\boldsymbol{P})),s)\\
        A_x(a,b,c,d,w,s\cup\{(v,\boldsymbol{L})\}) &= A_y(a-1,b,c,d,w,s)\\
        A_x(a,b,c,d,w,s\cup\{(v,\boldsymbol{R})\}) &= A_y(a-1,b,c,d,w,s)
    \end{align*}
    \item \textbf{Introduce edge bag:}
    \begin{itemize}
    \item If $s(u)=\boldsymbol{L}\wedge s(v)=\boldsymbol{R}$ 
    \begin{align*}
        A_x(a,b,c,d,w,s) = A_y(a,b,c,d-1,w-\omega((u,v),\boldsymbol{M}_2))
    \end{align*}
    \item If $s(u)=\boldsymbol{F}\lor s(v)=\boldsymbol{F}\lor s(u)=s(v)=\boldsymbol{R}$
    \begin{align*}
            A_x(a,b,c,d,w,s) &= A_y(a,b-[s(u)=s(v)\ne \boldsymbol{F}],c,d,w,s)
    \end{align*}
    \item  If $s(u)=s(v)=\boldsymbol{L}$
    \begin{align*}
    A_x(a,b,c,d,w,s) &= A_y(a,b-1,c,d,w,s)+A_y(a,b,c,d-1,w-\omega((u,v),\boldsymbol{M}_2),s)
    \end{align*}
        \end{itemize}
    Here we remove table entries not consistent with the edge $(u,v)$, and update the accumulator $b$ storing the number of edges in the induced subgraph and we mark the edge $(u,v)$ keeping $u$ in $X_L$ updating the accumulator $d$(even in the case when $u$ and $v$ are in $X_L$ and $X_R$ respectively) of edges in the induced subgraph.

    \item \textbf{Forget vertex bag:}
    \begin{align*}
        A_x(a,b,c,d,w,s) &= A_y(a,b,c-1,d,w-\omega((v,\boldsymbol{M}_1)),s[v\rightarrow \boldsymbol{L}])\\ &\;\;\;\;+ \sum\limits_{\alpha\in\{\boldsymbol{F},\boldsymbol{L},\boldsymbol{R}\}}A_y(a,b,c,d,w,s[v\rightarrow \alpha])
    \end{align*}
    If the vertex $v$ was in $X_L$ then we can mark it and update the accumulator $c$. If we do not mark the vertex $v$ then it can have any of the three states with no additional requirements imposed.\\
    \item \textbf{Join bag:}
    \begin{align*}
        A_x(a,b,c,d,w,s)&=  \sum\limits_{\mathclap{\substack{a_1+a_2=a+|s^{-1}(\{L,R\})|\\ b_1+b_2=b \\ c_1+c_2=c \\ d_1+d_2=d \\ w_1+w_2=w+\omega(s^{-1}(\boldsymbol{F})\times \{\boldsymbol{P}\})}}} A_y(a_1,b_1,c_1,d_1,w_1,s) \cdot A_z(a_2,b_2,c_2,d_2,w_2,s)
    \end{align*} 
    The only valid combinations to achieve the colouring $s$ is the same colouring in both the children bags. Since the vertices coloured $\boldsymbol{F}$ according to $s$ are present in both $y$ and $z$, their contribution to the weight $w$ and the number of the vertices $a$ needs to be accounted for.
\end{itemize}

 Since $|\mathcal{C}^{A,B,C,D}_W|\equiv A_r(A,B,C,D,W,\varnothing) \mod 2$, we compute $A_r(A,B,C,D,W,\varnothing)$ for all reasonable values of the parameters as mentioned before using the dynamic programming procedure, which takes $\mathcal{O}^\star(3^{\tw}|V|^{\mathcal{O}(1)})$ time. This concludes the description of the \emph{Cut \& Count} algorithm for {\sf pds}.

We state the following equivalent of Lemma \ref{lem:afdccmain}. The proof is omitted as it is very similar to the equivalent proof given for {\sc RIAFD}. 
\begin{lemma}
    \label{lem:pdsccmain}
    Let $G(V,E)$ be a graph and $d$ be an integer. Set the universe $U=V\times \{\boldsymbol{P},\boldsymbol{M}_1\}\cup E\times \{\boldsymbol{M}_2\}$. Pick $\omega'(u) \in \{1, \ldots, 2 |U|\}$ uniformly and independent at random for every $u \in U$. Define $\omega : U \rightarrow \mathbb{N}$ such that $\omega((v,\boldsymbol{P})):= |V|^2\omega'((v,\boldsymbol{P})) + deg(v)$ for all $v \in V$ and $\omega(u) = |V|^2\omega'(u)$ for all other $u \in U$. The following statements hold:                                                                                 
    \begin{enumerate}
        \item If for some integers $m'$, $D$, $W = i|V|^2 + d$ we have that $|\mathcal{C}_W^{n-k,m',n-k-m'-D,D}| \not\equiv 0\ (\textrm{mod} \ 2)$, then $G$ has a Pseudoforest Deletion set $P$ of size $k$ satisfying $deg(F) = d$.
        
        \item If $G$ has a Pseudoforest Deletion set $P$ of size $k$ satisfying $deg(P) = d$, then with probability at least $1/2$ for some $m'$, $D$, $W = i|V|^2 + d$ we have that $|\mathcal{C}_W^{n-k,m',n-k-m'-D,D}| \not\equiv 0\ (\textrm{mod}  \ 2)$.
    \end{enumerate}
\end{lemma}

\subsection{$\mathcal{O}^\star(3^k)$ Algorithm in Polynomial Space}\label{sec:pseudo_3k}
In this section, we present an $\mathcal{O}^\star(3^k)$ algorithm using polynomial space for solving \textsc{Pseudoforest Deletion}. First, we state the equivalent of Claim \ref{claim:afdpoly} and Theorem \ref{thm:afdcc} for \textsc{Pseudoforest Deletion} problem. Their proofs are omitted since they work out by replacing the \emph{Cut \& Count} algorithm for {\sc RIAFD} with \emph{Cut \& Count} for \textsc{PDS} described above, replacing \texttt{RIAFD-FCCount} with \texttt{PF-FCCount}, taking modulo with $2$ instead of $2^t$ and following a similar line of reasoning.

\begin{claim}\label{claim-pds}
Given a tree decomposition $\mathbb{T}$ with a set $S\subseteq V$ which is present in all its bags and a vertex assignment function $f: S \rightarrow \{\boldsymbol{F},\boldsymbol{L},\boldsymbol{R}\}$, there is a routine \texttt{PF-FCCount}$(\mathbb{T},R,A,B,C,D,W,f)$ which can compute $|\mathcal{C}_{W}^{A,B,C,D}| \;(\textrm{mod} \; 2)$ in time $\mathcal{O}^\star\left(3^{\tw - |S|}\right)$.
\end{claim}

\begin{theorem}\label{pds-space}
Given a tree decomposition $\mathbb{T}$, a set $S\subseteq V$ present in all bags of $\mathbb{T}$, parameter $k$, $\mathtt{CutandCountPF}$ solves the {\sc pseudoforest deletion} problem in $\mathcal{O}^{\star}\left(3^{\tw}\right)$ time and $\mathcal{O}^{\star}\left(3^{\tw - |S|}\right)$ space with high probability.
\end{theorem}

\begin{lemma}\label{lem:pseudo_ktw}
    Given a graph $G(V,E)$ and a {\sf pds} $P$ of size $k$, you can construct a tree decomposition $\mathbb{T}$ which contains the set $P$ in all bags and has width at most $k+2$ in polynomial time.
\end{lemma} 
\begin{proof}
    $G[V \setminus P]$ is a pseudoforest. Let $G[V \setminus F]$ have $c$ connected components. Let us consider the $i^{\text{th}}$ component $C_i$ and denote their individual tree decomposition as $\mathbb{T}_i$. $C_i$ is either a tree or a pseudoforest. If $C_i$ is a tree there is a trivial tree decomposition $\mathbb{T}_i$ of width $1$. If not, then $C_i$ is a pseudotree. Remove any edge $(u,v)$ from the only cycle in $C_i$ and construct the tree decomposition of the remaining tree. Add the vertex $u$ in all bags of that tree decomposition to get $\mathbb{T}_i$ of width $2$ for the pseudotree $C_i$. Now, make an empty bag as the root and connect the root of all $\mathbb{T}_i$ to it and call the resulting tree decomposition (of width $2$) $\mathbb{T}'_i$. Now, adding $P$ to all bags of $\mathbb{T}'_i$ gives the desired tree decomposition $\mathbb{T}_i$ of width $k+2$. The time bound is trivial from the description of the procedure. 
\end{proof}

Now, we state the following lemma and prove it.

\begin{lemma}\label{lem:pseudo_3k3l}
There exists an algorithm $\mathtt{PF3k}$ that solves \textsc{Pseudoforest Deletion} in $\mathcal{O}^\star(3^k)$ time and polynomial space with high probability.
\end{lemma}

\begin{proof}
In algorithm \texttt{RIAFD3k3l} from Section \ref{sec:afd_3k3l}, replace \texttt{RIAFDCutandCount} with \texttt{CutandCountPF}. Also replace the equivalent lemmas, theorems and claims. Denote this algorithm as \texttt{PF3k}. The proof of correctness and success-probability is similar to Theorem \ref{thm:afd} $(1)$ in Section \ref{sec:afd_3k3l}. The running time and space bound follow by similar arguments in the proof of Theorem \ref{thm:afd} $(1)$ and Lemmas \ref{pds-space}, \ref{lem:pseudo_ktw}.
\end{proof}

\subsection{$\mathcal{O}^{\star}(2.85^k)$ Algorithm in Polynomial space}\label{sec:pseudo_improvedk}
In this section, we present a $\mathcal{O}^{\star}(2.85^k)$ algorithm using polynomial space. We use the method from \cite{li-soda20}, dividing the problem into sparse and dense cases. Following are a few basic reduction rules for \textsc{Pseudoforest Deletion}, which are quite similar to those for \fvs.

\begin{definition}\label{red:pseudo1}
    \textbf{Reduction $1$:} Apply the following reduction rules exhaustively until there is no edge of multiplicity larger than 3, there are vertices with at most one loop and degree at least 3.
    \begin{enumerate}
        \item If there is more than one loop at a vertex $v$, delete v and decrease $k$ by 1; add $v$ to the output {\sf pds}.
        \item If there is an edge of multiplicity larger than 3, reduce its multiplicity to 3.
        \item If there is a a vertex $v$ of degree at most 1, delete $v$.
        \item If there is a vertex  $v$ of degree 2, delete $v$ and join its neighbours with an edge.
        \item If $k<0$, then we have a no instance. If $k>0$ and $G$ is a pseudoforest, then we have a yes instance. If $k=0$, we have a yes instance iff $G$ is a pseudoforest.
    \end{enumerate}
\end{definition}

\subsubsection{Dense case}
In this case, we apply a probabilistic reduction that capitalises on the fact that a large number of edges are incident to the {\sf pds}. We will use the same ideas as of Reduction $2$ for {\sc RIAFD} in Section \ref{sec:afd_improve_k}. Thus, even here we aim to obtain a reduction that succeeds with probability strictly greater than $1/3$ so as to achieve a randomized algorithm running in $\mathcal{O}^{\star}(3-\epsilon)^k$ time that succeeds with high probability.

\begin{definition} \label{red:pseudo2}
    \textbf{Reduction $2$ (P):} Assume that Reduction $1$ does not apply and $G$ has a vertex of degree at least $3$. Sample a vertex $v \in V$ proportional to $\omega(v) := (deg(v)-2)$. That is, select each vertex with probability $\frac{\omega(v)}{\omega(V)}$. Delete $v$ and decrease $k$ by $1$.
\end{definition}

\begin{claim} \label{claim:pseudo_degf}
    Let $G$ be a graph, $P$ a {\sf pds} of $G$. Denote $\overline{P} := V \setminus P$. We have that,
    $$deg(\overline{P}) \le deg(P) + 2(|\overline{P}|). $$
\end{claim}

\begin{lemma} \label{lem:pseudo_red2}
    Given a graph $G$, if there exists a {\sf pds} $P$ of size $k$ such that $deg(P) \ge \frac{4-2\epsilon}{1-\epsilon}k$, then success of Reduction $2$ which is essentially picking a vertex $v$ from the {\sf pds} $P$ occurs with probability at least $\frac{1}{3-\epsilon}$.
\end{lemma}

The proofs of the above claim and lemma follow a similar line of reasoning as the proofs of Claim \ref{claim:afd_degf} and Lemma \ref{lem:afd_red2}, hence they are omitted.

\subsubsection{Sparse case}
In this case, since $deg(P)/|P| \le \overline{d}$ and $\overline{d} = \mathcal{O}(1)$, it is possible to get a tree decomposition of size $(1-\Omega(1))k$.

We state this without proof through the following lemmas since they use the same ideas from \cite{li-soda20}.

\begin{lemma}\label{pds_sep}
Given $(G, k)$ and a {\sf pds} $P$ of $G$ of size exactly $k$, define $\overline{d} := \frac{deg(P)}{k}$ , and suppose that $\overline{d} = \mathcal{O}(1)$. There is a randomized algorithm running in expected polynomial time that computes a separation $(A, B, S)$ of $G$ such that:
\begin{enumerate}
    \item $|A\cap P|,|B\cap P|\geq (2^{-\overline{d}}-o(1))k$
    \item $|S|\leq (1+o(1))k-|A\cap P|-|B\cap P|$
\end{enumerate}
\end{lemma}

\begin{proof}
    The proof is similar to that in \cite{li-soda20}. The only difference is in the first step i.e construction of a $\beta$-separator $S_{\epsilon}$. For this we can use Lemma \ref{lem:generalized_beta_separator} which gives a $\beta$-separator of size at most $3\beta$ ($\tw$ of any pseudoforest is $\le 2$), and as $\beta = \epsilon k = o(k)$, $|S_{\epsilon}| = o(k)$. All other steps and bounds remain exactly the same.
\end{proof}

\begin{lemma}\label{pds_tw}
Let $G$ be a graph and $P$ be a {\sf pds} of $G$ of size $k$, and define $\overline{d} := \frac{deg(P)}{k}$. There is an algorithm that, given $G$ and $P$, computes a tree decomposition of $G$ of width at most $(1-2^{-\overline{d}}+o(1))k$, and runs in polynomial time in expectation.
\end{lemma}

As we are in the sparse case, which means that there exists a {\sf pds} $P$ of size $k$ with bounded degree, i.e., $deg(P) \le \overline{d}k$. We call this bounded version of the problem, {\sc BPDS}. As we saw, the small separator helps in constructing a tree decomposition of small width, but requires that we are given a {\sf pds} of size $k$ and bounded degree. To attain this, we use an Iterative Compression based procedure which at every iteration considers a {\sf pds} of size at most $k$ with bounded degree and uses it to construct the small separator. Using this small separator we construct a tree decomposition of small width and employ a \emph{Cut \& Count} based procedure to solve {\sc BPDS} for the current induced subgraph, i.e, get a {\sf bpds} of size at most $k$ with bounded degree. This {\sf bpds} is used for the next iteration, and so on. 

% \begin{lemma}\label{deg_order_PDS}
% For graph $G=(V,E)$ with $|V|=n$, order its vertices in the non-increasing order of their degrees. Let $G_i:=G[\{v_1,\ldots,v_n\}]$. If $G$ has a $PDS$ $P$ of size at most $k$ with $deg(P)/|P|\leq \overline{d}$, then $G_i$ has a $PDS$ $P'$ of size at most $k$ with $deg(P')/|P'|\leq \overline{d}$ for each $i$.
% \end{lemma}

% \begin{proof}
% Let $V_i$ denote the vertex set for $G_i$. Clearly $P''=P\cap V_i$ is a $PDS$ for $G_i$ of size at most $k$. Also, $deg(P'')/|P''|\leq deg(P)/|P|=\overline{d}$ because any vertex in $P-P''$ cannot have a lower degree than all of the vertices in $P''$ due to the non-increasing degree order of vertices.
% \end{proof}

\begin{note}
\label{note_tree_dec}
Using the tree decomposition obtained in \ref{pds_tw}, we can run the \emph{Cut \& Count} algorithm from Section \ref{sec:pseudo_3tw}. But this will utilize exponential space. To get polynomial space, we use the following idea.

Given an $(A,B,S)$ separation of a graph $G$ according to Lemma \ref{pds_sep} along with a {\sf pds} $P$ of size at most $k$ of bounded average degree $\overline{d}$, we construct a tree decomposition  $\mathbb{T}'$ of $G$ as follows: Since $A\cap P\cup S$ is a {\sf pds} for $A\cup S$, we construct a nice tree decomposition $\mathbb{T}_1$ of $A\cup S$ which forgets all vertices in $S$ at the last(going from a leaf bag to the root bag). Hence there is a bag $B_y$ which contains all the vertices $v\in S$ and nothing else. Upto this bag, no edge $e\in E[S,S]$ is introduced. Consider this part of the tree decomposition of $\mathbb{T}_1$(denote as $\mathbb{T}_1'$) up to node $y$. Similarly we construct a tree decomposition $\mathbb{T}_2$ for partition $B$ as we did for $A$. There is a bag $B_z$ in $\mathbb{T}_2$ which contains all vertices $v\in S$ and nothing more. Denote the tree decomposition up to node $z$ for $\mathbb{T}_2$ as $\mathbb{T}_2'$. The final tree decomposition $\mathbb{T}$ for $G$ is constructed by joining $\mathbb{T}_1'$ and $\mathbb{T}_2'$ via a join node and then going toward the root we have the introduce edge bags and forget vertex bags for $v\in S$. We use this tree decomposition $\mathbb{T}'$ for proving the polynomial space bound in Algorithm 10. 
\end{note}

% \begin{algorithm}
% \SetKwInput{KwInput}{Input}
% \SetKwInput{KwOutput}{Output}
% \SetKwComment{Comment}{$\triangleright$ Invariant:}{}
% \SetAlgoLined
% \KwInput{Graph $G=(V,E)$ and parameters $k\leq n$ and $\overline{d}=O(1)$}
% \KwOutput{A PDS $P$ of size at most $k$ satisfying $deg(P)/|P|\leq \overline{d}$ or $\mathtt{Infeasible}$ if none exists}.
% Order the vertices $V$ in the non-increasing order of their degrees as $(v_1,\ldots,v_n)$\\
% $P\leftarrow \varnothing$\\

% \For{$i=1,\ldots,n$}{
% Compute a separation $(A,B,S')$ of $G[\{v_1,\ldots,v_{i-1}\}]$ \Comment*[f]{$deg(P)/|P|\leq \overline{d}$}
% by Lemma \ref{pds_sep} on input $P$\\
% $S\leftarrow S'\cup v_i$ so that $(A,B,S)$ is a separation of $G[\{v_1,\ldots,v_i\}]$\\
% $P\leftarrow \mathtt{BPDS}(G[\{v_1,\ldots,v_i\}],k+1,A,B,S)$\\
% \If{$P$ is $\mathtt{Infeasible}$}{
% \Return{$\mathtt{Infeasible}$}}
% }
% \Return{$P$}
% \caption{$\mathtt{PFIC1}(G,k,\overline{d})$}
% \end{algorithm}

% We now give the algorithm for $\mathtt{BPDS}$ that essentially solves the problem of finding a bounded average degree $PDS$ of size at most $k-1$ given $G$, $PDS$ of size at most $k$ the parameter $\overline{d}$ and separation $(A,B,S)$, after which we state the lemma regarding correctness of Algorithm $\mathtt{PFIC1}(G,k,\overline{d})$.

Now, we give the claimed \hyperref[alg:bpds]{\texttt{BPDS}} algorithm, which is a \emph{Cut \& Count} based algorithm which solves bounded degree {\sc PDS} given a small separator.

\begin{algorithm}[!ht]
\label{alg:bpds}
\SetKwInput{KwInput}{Input}
\SetKwInput{KwOutput}{Output}
\caption{$\mathtt{BPDS}(G,P,k,A,B,S)$}
\SetAlgoLined
\KwInput{Graph $G=(V,E)$, {\sf pds} $P$ of size at most $k+1$, parameters k, $\overline{d}\leq n$ and a separation $(A,B,S)$ from Lemma \ref{pds_sep}.}
\KwOutput{A {\sf pds} $P$ of size at most $k$ satisfying $deg(P)/|P|\leq \overline{d}$ or $\mathtt{Infeasible}$ if no such set exists.}
\Begin{
Set the universe $U=V\times \{\boldsymbol{P},\boldsymbol{M}_1\}\cup E\times \{\boldsymbol{M}_2\}$\\
Pick $\omega$ uniformly and independently at random as defined in Lemma \ref{lem:pdsccmain}

%$$
%\omega'(u):=
%\begin{cases}
%|V|^{2}\omega(u)+d(u)  &\quad            \text{for $u\in \{V\times \boldsymbol{P}$\}}\\
%|V|^{2}\omega(u)       &\quad      \text{Otherwise}\\
%\end{cases}
%$$\\
Construct the tree decomposition $\mathbb{T}'$ as stated earlier in note \ref{note_tree_dec}\\
Compute $C_W^{A,B,C,D}$ for all reasonable values of A,B,C,D,W using $\mathtt{CutandCountPF}$\\
\Return{A {\sf pds} $P$ with $|P|\leq k$ and $deg(P)/|P|\leq \overline{d}$}
}
\end{algorithm}

\begin{lemma}\label{BPDS}
There is an Algorithm $\mathtt{BPDS}$ that, given $G$, a {\sf pds} $P$ of $G$ of size at most $k+1$, parameter $\overline{d}$, and a separation $(A, B, S)$ as given by Lemma \ref{pds_sep}, outputs a {\sf pds} of size at most $k$ satisfying $deg(P)/|P| \leq d$,
or $\mathtt{Infeasible}$ if none exists. The algorithm uses $\mathcal{O}^{\star}(3^{(1-2^{-\overline{d}}+o(1))k})$ time and polynomial space.
\end{lemma}

\begin{proof}
Note that we reorder the computation of algorithm \texttt{CutandCountPF} in a slightly different way on tree decomposition $\mathbb{T}'$ to achieve polynomial space. Follow the notations according to note \ref{note_tree_dec}. The way we reorder the computation of \texttt{CutandCountPF} on tree decomposition $\mathbb{T}'$ is as follows: For a fixed colouring s of $S$, we compute $A_y(a,b,c,d,w,s)$ and $A_z(a,b,c,d,w,s)$ in polynomial space according to Claim \ref{claim-pds}. Now the remaining tree decomposition has bags only consisting of vertices in $S$. Using $A_y(a,b,c,d,w,s)$ and $A_z(a,b,c,d,w,s)$ for some colouring s of $S$ we can compute $C_W^{A,B,C,D}$ for $\mathbb{T}'$ in polynomial space by Theorem \ref{pds-space}.\\
The algorithm is clearly correct since it uses \texttt{CutandCountPF} as a subroutine with reordered computation. By Lemma \ref{lem:pdsccmain}, the {\sf pds} $P$ of size at most $k$ is found using \texttt{CutandCountPF} with bounded average degree $\overline{d}$ with success probability at least $1/2$. The success probability can be easily boosted by $n^{\mathcal{O}(1)}$ runs of the algorithm. The width of the tree decomposition from the input according to Lemma \ref{pds_tw} is ${(1-2^{-\overline{d}}+o(1))k}$. Thus the time bound follows the time bound of the $\mathtt{CutandCountPF}$ algorithm.
\end{proof}

% We state the following lemma without proof because it is similar to the proof of the same statement of Algorithm $\mathtt{IC2}$ in \cite{li-soda20}.

% Only thing to note here is that due to Lemma \ref{deg_order_PDS} the non-increasing degree arrangement ensures the invariant at each step of the existence of an average degree $PDS$(assuming that a $PDS$ of bounded average degree exists for the whole graph)

Now, we give the Iterative Compression routine which solves \textsc{BPDS}, as explained above.

\begin{lemma}\label{PFIC2}
There exists an algorithm $\mathtt{PFIC1}$ that solves \textsc{BPDS} in $O^{\star}(3^{(1-2^{-\overline{d}}+o(1))k})$ time and polynomial space with high probability.
\end{lemma}
\begin{proof}
    \texttt{PFIC1} can be constructed by replacing every occurrence of \texttt{BRIAFD1} with \texttt{BPDS} and constructing the separator using Lemma \ref{pds_sep}. The proofs of correctness, space bound and success-probability are similar to Lemma \ref{lem:afd_IC1}.
\end{proof}

\subsubsection{Combining Sparse and Dense Cases}
Having described the \emph{Dense} and the \emph{Sparse} Cases, we now combine them to give the final randomized algorithm.

\begin{lemma}\label{PDS1}
Fix the parameter $\epsilon\in (0,1)$ and let $c_\epsilon:=max\left\{3-\epsilon,3^{1-2^{-\frac{4-2\epsilon}{1-\epsilon}}}\right\}$. If $c_\epsilon \ge 2$, there exists an algorithm $\mathtt{PDS1}$ that succeeds with probability at least $c_\epsilon^{-k}$. Moreover Algorithm $\mathtt{PDS1}$ has expected polynomial running time and requires polynomial space.
\end{lemma}
\begin{proof}
    In algorithm \texttt{RIAFD1}, replace every occurrence of \texttt{RIAFDIC1} with \texttt{PFIC1}. Also, replace the Reduction rules with the ones given for Pseudoforest Deletion. This modified algorithm is \texttt{PDS1}. The running time, space bound and success probability analysis are similar to the analysis in proof of Lemma \ref{lem:afd_AFD1}.
\end{proof}

Note that the outer loop on $k$ is not required here. If there exists a {\sf pds} of size at most $k$, we can add arbitrary vertices to get a {\sf pds} of size exactly $k$.

To optimize for $c_\epsilon$, we set $\epsilon\approx 0.155433$, giving $c_\epsilon\approx 2.8446$. Using Lemma \ref{lem:randomalgo} we can boost the success probability to be sufficiently high. Theorem \ref{thm:pseudo} thus follows from Lemma \ref{PDS1} and \ref{lem:randomalgo}.

 \section{Conclusion}
 %!TEX root = main.tex
In this paper, we applied the technique of Li and Nederlof~\cite{li-soda20} to other problems around 
the {\sc Feedback Vertex Set} problem. The technique of Li and Nederlof is inherently 
randomized, and it uses the \emph{Cut \& Count} technique, which is also randomized.  
Designing matching deterministic algorithms for these problems, as well as for {\sc Feedback Vertex Set}, is a long standing open problem. However, there is a deterministic algorithm for {\sc Pseudoforest Deletion} running in time ${\cal O}^{\star}(3^k)$~\cite{BodlaenderOO18}. So obtaining a deterministic algorithm for {\sc Pseudoforest Deletion} running in time $\mathcal{O}^{\star}(c^k)$ for a constant $c<3$ is an interesting open question. Further, can we design an algorithm for {\sc Pseudoforest Deletion} running in time $\mathcal{O}^{\star}(2.7^k)$, by designing a different \emph{Cut \& Count} based algorithm for this problem? Finally, could we get a ${\cal O}^{\star}(c^k 2^{o(\ell)})$ algorithm for {\sc Almost Forest Deletion}, for a constant $c$ possibly less than $3$? 

\newpage
\bibliographystyle{plainurl}
%\bibliography{references.bib,references-saket.bib}
\bibliography{reference,ref,references-cfl}
% 

%\section*{References}
%\printbibliography[heading=none]
\end{document}